\documentclass[draftclsnofoot,onecolumn]{IEEEtran}

\usepackage{algorithm}
\usepackage[noend]{algpseudocode}	

\usepackage[hidelinks]{hyperref}

\newcommand{\zz}{\mathbb{Z}^+}
\usepackage{amsmath}
\usepackage{amssymb,amsthm}
\usepackage{amsfonts}
\usepackage{cite}
\usepackage{tikz}
\usetikzlibrary{positioning,arrows,shapes,chains,fit,scopes,calc}
\usepackage{pgfplots}
\usepackage{pgfplotstable}
\usepackage{booktabs}
\usepackage{colortbl}
\usepackage{subfigure}
\usepackage{tabularx}
\usepackage{booktabs}
\usepackage[normalem]{ulem}
\usepackage[draft]{changes}

\usepackage{blkarray}

\theoremstyle{theorem}
\newtheorem{theorem}{Theorem}
\newtheorem{lemma}{Lemma}

\newtheorem{proposition}{Proposition}
\theoremstyle{definition}
\newtheorem{definition}{Definition}
\newtheorem{example}{Example}
\theoremstyle{remark}
\newtheorem{remark}{Remark}

\newcommand{\tx}{\mathrm{s}}
\newcommand{\rx}{\mathrm{r}}

\newcommand{\conv}{\mathrm{conv}}
\newcommand{\cl}{\mathrm{cl}}
\newcommand{\dom}{\mathrm{dom}}

\newcommand{\mc}[1]{\mathcal{#1}}
\newcommand{\cv}[1]{\mathbf{#1}}
\newcommand{\cR}{{\mathcal{R}}}
\newcommand{\dR}{\widetilde{\mathcal{R}}}

\newcommand{\pleq}{\preccurlyeq}
\newcommand{\pgeq}{\succcurlyeq}

\tikzstyle{src}=[circle,draw=gray!80!black,fill=gray!10,thick,inner
sep=1pt,minimum size=12pt]
\tikzstyle{dst}=[circle,draw=gray!80!black,fill=gray!10,thick,inner
sep=1pt,minimum size=12pt]
\tikzstyle{mid}=[circle,draw=gray!80,fill=gray!10,thick,inner
sep=1pt,minimum size=12pt]

\title{Wireless Network Scheduling with Discrete Propagation Delays: Theorems and Algorithms}
 \author{\IEEEauthorblockN{Shenghao~Yang,~Jun~Ma~and~Yanxiao~Liu}%
\thanks{This paper was presented in part at INFOCOM 2021~\cite{jma21scheduling}.}
\thanks{S.~Yang is with the School of Science and Engineering, The Chinese University of Hong Kong, Shenzhen, Shenzhen, China. Email: shyang@cuhk.edu.cn}
\thanks{J.~Ma is with QIANHAI FOF, Shenzhen, China, and with the School of Management, Xiamen University, Xiamen, China. Email: jma@qhfof.com}
\thanks{Y.~Liu is with the Department of Information Engineering, The Chinese University of Hong Kong, Hong Kong, China. Email: yanxiaoliu@link.cuhk.edu.hk}
\thanks{This work was funded in part by the Shenzhen Science and Technology Innovation Committee (Grant JCYJ20180508162604311).}}

\begin{document}

\maketitle
\begin{abstract}
  The literature provides evidence that considering signal propagation
  delays can significantly enhance the scheduling rate region of
  wireless networks. This paper focuses on the link scheduling problem
  in networks where signal delays between nodes are multiples of a
  time interval. To model such networks, a directed hypergraph is
  employed, along with an integer matrix that specifies the delays.
  The link scheduling problem is closely connected to the independent
  sets of the periodic hypergraph induced by the network
  model. However, due to the infinite number of vertices, it is
  impractical to enumerate the independent sets of the periodic
  hypergraph using generic graph algorithms. To tackle this challenge,
  a graphical approach is proposed in this paper. The link scheduling
  rate region is characterized using a finite directed graph called a
  scheduling graph, which is derived from the network model. A
  collision-free schedule of the network corresponds to a path in the
  scheduling graph, and the rate region is determined by the convex
  hull of the rate vectors associated with the cycles in the
  scheduling graph. Although existing cycle enumeration algorithms can
  be employed to calculate the rate region, their computational
  complexity becomes prohibitively high as the size of the scheduling
  graph grows exponentially with the number of network links. To
  address this issue, the dominance property of a special scheduling
  graph called the step-$T$ scheduling graph is investigated. This
  property allows the utilization of specific subgraphs of the
  step-$T$ scheduling graph to characterize the scheduling rate
  region, achieving a reduction in both the number of cycles and their
  lengths. For common problems such as calculating the rate region and
  maximizing a weighted sum of the scheduling rates, algorithms
  leveraging the dominance property are developed.  These algorithms
  can be more efficient than using generic graph algorithms
  directly on the scheduling graphs.
\end{abstract}

\section{Introduction}

Wireless communication media, such as radio, light, and sound, all
have nonzero signal propagation delays between two communicating
devices separated by a nonzero distance. In the existing theory of
wireless communications, these signal propagation delays are typically
regarded as a factor that can potentially generate
interference~\cite{proakis2003intersymbol}. However, studies on
underwater acoustic communications and interference channels have
observed that wireless communication networks can actually benefit
from these signal propagation
delays~\cite{hsu2009st,guan2011mac,grokop2011interference,
  chitre2012throughput,anjangi16,bai2017throughput,ma2019hybrid}.  To
gain a better understanding of the utilization of propagation delays,
we present a theoretical framework for studying network scheduling
when considering these signal propagation delays.

\subsection{Background and Related Results}

In most terrestrial radio-based wireless communication systems, such
as the 5G cellular network, guard intervals are employed in network
scheduling to mitigate the interference caused by signal propagation
delays.  The inclusion of guard intervals does not significantly
degrade system performance since the duration for transmitting signals
is much greater than the signal propagation delay between devices.
For communication within a few kilometers, the propagation delay of
radio waves typically falls within the range of tens of
microseconds. In contrast, the signal frame length is typically a few
milliseconds. This scheduling approach, characterized by a long signal
frame length, is referred to as \emph{framed scheduling} and
constitutes a focal point of research in wireless network scheduling
for terrestrial radio-based
communications~\cite{hajek1988link,ephremides1990scheduling,
  tassiulas92stability,sharma2006complexity}.

In underwater acoustic communications, the propagation delay of sound
can be significantly longer, measured in seconds. For instance, the
sound speed in underwater environments is approximately $1,500$ meters
per second, resulting in a delay of around $2$ seconds for sound to
propagate over a distance of $3$ kilometers.  If framed scheduling is
adopted in this scenario, the frame length should be on the order of
tens or hundreds of seconds~\cite{le2013pmac}.  Researchers in
underwater acoustic networks have been motivated to address delays in
the network scheduling
problem~\cite{hsu2009st,guan2011mac,chitre2012throughput,anjangi16,bai2017throughput,ma2019hybrid}.
They have observed substantial performance advantages, such as
improved energy consumption and throughput, by allowing smaller frame
lengths that are comparable to the propagation delay between
communication devices.

Researchers have demonstrated that in a network consisting of $K$ pairs
of closely located communication devices, framed scheduling allows
only one device to transmit a signal at a given timeslot without
collision. By carefully considering the delays, it is feasible to
devise scheduling schemes where all $K$ pairs can transmit
simultaneously without generating collisions \cite{chitre2012throughput}. Moreover, recent work~\cite{fan22isit} has provided examples
with relaxed delay constraints to illustrate the unbounded advantages
of scheduling with propagation delays, particularly when the network
size is large.

The advantage of utilizing propagation delays has also been discovered in the study of the time-domain interference alignment approach for multi-user interference channels~\cite{grokop2011interference}. 
Though the time-domain interference alignment and network
scheduling are not equivalent problems,\footnote{For instance, in time-domain interference alignment, the timeslot size corresponds to the sampling time interval, and each timeslot contains a single sample value~\cite{grokop2011interference}. In contrast, in the network scheduling problem, each timeslot is typically considered as a radio frame containing a sequence of sample values.} it is possible to transform a network schedule
into an achievable scheme for the interference
channel~\cite{grokop2011interference}.  Using this approach, it is shown
in~\cite{grokop2011interference} that there exists a non-vanishing rate
for each user when the number of users tends to infinity. Additionally, some papers
discussing underwater acoustic networks also refer to scheduling with
propagation delays as interference
alignment~\cite{zeng2014shark,zhao2021message}.

\begin{table*}[tb]
  \centering
  \caption{Comparison of the signal propagation delay and the OFDM symbol length in both underwater acoustic communication and territory radio communication.
     The OFDM subcarrier spacing for underwater acoustic communication refers to~\cite{wan2015ofdm}.
    The $15$ kHz OFDM subcarrier spacing is used in both 4G LTE and 5G NR, and the $15\times 2^6$ kHz OFDM subcarrier spacing is used in 5G NR. The OFDM symbol length does not include the cyclic prefix. 
  }
  \label{tab:1}
  \begin{tabular}{lrrrrr}
    \toprule
    & propagation & transmission  & propagation  & OFDM subcarrier  & OFDM symbol  \\
    & speed (km/s) & range (km)      & delay (s) & spacing (kHz) & length (s) \\
    \midrule
    underwater acoustic & $1.5$ & $3$ & $2$ & $0.005$  & $0.2$ \\
    4G LTE & $3\times 10^5$ & $3$ & $1 \times 10^{-5}$ & $15$ & $6.67\times 10^{-5}$ \\
    5G NR & $3\times 10^5$ & $3$ & $1 \times 10^{-5}$  & $15\times 2^{6}$  & $1.04\times 10^{-6}$ \\
    \bottomrule
  \end{tabular}
\end{table*}

The benefit of using propagation delays can be achieved in terrestrial
radio communications when a sufficiently large bandwidth is
utilized. We can consider orthogonal frequency-division multiplexing
(OFDM) as an example, which is employed in many modern wireless
communication systems. In OFDM, a frame is typically composed of
multiple OFDM symbols. A comparison of OFDM numerology in different
systems is presented in Table~\ref{tab:1}.  In underwater acoustic
communications, the OFDM symbol length can often be much shorter than
the typical signal propagation delay. However, for 4G LTE with a $15$
kHz OFDM subcarrier spacing, the OFDM symbol length is more than $6$
times the typical signal propagation delay. In 5G NR, larger OFDM
subcarrier spacings up to $15\times 2^6$ kHz are proposed for
bandwidths up to $400$ MHz in millimeter-wave frequencies. In this
case, the OFDM symbol length can be approximately $1/10$ of the
typical signal propagation delay.  Furthermore, wireless communication
in the frequency range of $100$ GHz to $10$ THz has been discussed for
the next generation of cellular networks~\cite{o2019perspective}. In
such scenarios, the bandwidth can reach tens of gigahertz, and the
OFDM symbol length can be several nanoseconds.

Although the study of wireless networks with propagation delays shows
promise, it is still in its preliminary stages. Early
works~\cite{hsu2009st,guan2011mac} have utilized mixed integer linear
programming models to capture collision constraints and derive
heuristic algorithms. In cases where delays are integers, the
scheduling problem with propagation delays is formulated as a weighted
directed graph. In this graph, the vertices represent communication
links, directed edges model collision relations between two links, and
the weight of an edge denotes the corresponding
delay~\cite{hsu2009st,grokop2011interference}.  For complete weighted
directed graphs (where any two links can generate collisions with each
other), existing works~\cite{grokop2011interference,
  chitre2012throughput} have discovered that the network scheduling
problem exhibits a periodic property. Dynamic programming approaches
have been employed to maximize the (weighted) total scheduling
rate. However, these approaches suffer from high computation costs due
to the exponential state space and do not provide an explicit
characterization of the scheduling rate region.

Without considering propagation delays, the network scheduling problem
can be formulated using a
graph or a
hypergraph~\cite{arikan1984some,hajek1988link,ephremides1990scheduling,tassiulas92stability,jain2003impact,sharma2006complexity,lin2006,wu2007scheduling,chaporkar2008throughput,mceliece1994performance,sarkar1998hypergraph,li2012maximal,zhang2016radio,ganesan2021some}.  In this formulation,
 a vertex represents a network link and an edge represent the
collision relation among the links. When considering a \emph{protocol or binary
  collision model}~\cite{hajek1988link,tassiulas92stability}, a graph can be used to capture the pair-wise collision
relation among the network links. In a more practical \emph{physical
  Signal-to-Interference-and-Noise Ratio (SINR)
  model}~\cite{tse2005fundamentals}, a group of links can collectively
generate a collision with another link, which is described using a
hypergraph. Although the network links can be directed, the graphical representation of collision is usually undirected.

The graphical model of collisions plays a crucial role in network scheduling research. It allows for the explicit characterization of the scheduling rate region using the independent sets of the graph or hypergraph, enabling the analysis of the wireless network's performance based on graphical properties.
However, when it comes to scheduling with delays, a graphical theory specifically addressing this aspect is currently lacking but highly desired. In this paper, our objective is to fill this gap by proposing a graphical approach that characterizes the scheduling rate region with delays.

\subsection{Main Contributions}

For a wireless network, if all the signal propagation delays are
multiples of a fixed timeslot length, we say the delays are
\emph{discrete}.  In this paper, we study wireless networks scheduling
with discrete propagation delays, which serves as an initial step
towards understanding networks with general real number delays.  For a
network with general delays, the rate region can be approximated by
the rate region of a network with discrete delays~\cite{fan22isit}.  A
general approach is also presented in~\cite{fan22isit} to approximate
a network with general delays using one that has discrete delays, with
a controlled performance difference. Additionally, the problem of
time-domain interference alignment can
also be approximated as a problem with discrete signal propagation
delays~\cite{grokop2011interference}.

We propose a comprehensive network model that incorporates a matrix to
represent delays and a \emph{directed hypergraph} to describe
collision relations. In the case of a binary collision model, the
directed hypergraph simplifies to a directed graph. When all delays
are zero, the model reduces to one without considering delays,
allowing for the omission of edge direction without impacting the
network scheduling problem. However, in the presence of general delays, edge
direction becomes crucial, and the network scheduling problem is
connected to independent sets within the \emph{periodic (undirected)
  hypergraph} induced by our network model.

Despite the connection between the network scheduling problem and
independent sets within the periodic hypergraph, finding a complete
solution remains elusive due to the infinite nature of the periodic
hypergraph. Notably, independent sets can have unbounded sizes, and
there exists an infinite number of them. Consequently, existing
approaches can only provide approximations of the scheduling rate
region, and the computational cost is high, as demonstrated in prior
research~\cite{grokop2011interference, chitre2012throughput}.  In this
paper, we address the challenges associated with the infinite number
of independent sets in the periodic hypergraph, and our main
contributions are twofold: i) exact and explicit characterizations of
the scheduling rate region, and ii) efficient algorithms for
calculating the rate region and maximizing a weighted sum of the link
rates.

We show that the scheduling rate region of a network can be achieved
using collision-free, periodic schedules. To provide an explicit
characterization of this scheduling rate region, we adopt a graphical
approach. For our network model, we establish a series of directed
graphs called \emph{scheduling graphs}. Each scheduling graph has two
parameters: a blocklength $T$ and a step size $Q$ (where
$1\leq Q\leq T$). We show in general that a collision-free schedule is
equivalent to a path within a scheduling graph with $T\geq 2D^*$,
where $D^*$ is a parameter derived from the delay
matrix. Consequently, the scheduling rate region is the convex hull of
the rate vectors associated with the cycles of the scheduling graph,
and hence is a polytope.\footnote{In this paper, a cycle in a directed
  graph has no repeated vertices (which is also called a simple
  circuit) and hence the total number of cycles of a finite graph is
  finite.} In the case of binary collision, the characterization of
the scheduling rate region can be achieved using a scheduling graph
with $T\geq D^*$.  It is worth noting that the scheduling graphs,
regardless of whether they are induced by graph-based or
hypergraph-based network models, share common properties, except for
the variation in the bounds imposed on $T$. These common properties
allow for a unified study of scheduling related problems based on the
scheduling graphs.

We further study scheduling-related algorithms, specifically focusing
on computing the rate region and maximizing a weighted sum of the link
rates. Based on the characterization of the scheduling rate
region, we explore various approaches to address these computational
problems. As a straightforward approach, one can employ a backtracking
algorithm, such as Johnson's algorithm~\cite{johnson1975finding}, to
enumerate cycles within a scheduling graph and consequently obtain the
rate region. However, the computational cost of this approach becomes
prohibitively high as the network size increases. This can be
attributed to two main factors: the exponential growth of vertices and
edges in a scheduling graph relative to the number of network links,
and the exponential growth in the number of cycles as
the size of the scheduling graph increases.

To simplify the characterization of the rate region, we investigate an
additional property of the scheduling problem. Specifically, we
introduce a dominance property for the step-$T$ scheduling graph,
where the step size $Q = T$. This property allows us to leverage
subgraphs of the step-$T$ scheduling graph to characterize the
scheduling rate region, achieving a reduction in both the number of
cycles and their respective lengths. To illustrate the benefits of the
dominance property in the step-$T$ scheduling graph, we provide an
example involving a sequence of step-$T$ scheduling graphs. In these
scheduling graphs, the numbers of edges and vertices are exponential in
the number of links, but the corresponding reduced scheduling graphs
possess a constant number of cycles with constant lengths.

Based on the dominance property, we develop two algorithms for
calculating the scheduling rate region. The first algorithm enables
the calculation of the entire rate region by enumerating only the
cycles present in a reduced scheduling graph. This approach proves to
be more efficient than enumerating all cycles in the original
scheduling graph. The second algorithm takes an incremental approach,
specifically targeting a subset of the scheduling rate region
characterized by cycles up to a certain length. Numerical evaluations
demonstrate that this algorithm outperforms the direct enumeration of
cycles up to a specific length, particularly in scenarios involving
large network sizes.

To solve a maximization problem on the scheduling rate region, the
straightforward approach involves two steps: computing the scheduling
rate region or a subset thereof, and then maximizing the objective
function within the feasible rate vectors obtained.  However, this
approach can become impractical for larger networks due to the substantial
computational cost associated with calculating the rate region. To
address this challenge, we propose an algorithm that leverages the
insights from the dominance property. This algorithm maximizes a
linear function without the need to compute the entire rate region,
resulting in significantly lower computation costs compared to the
straightforward approach.

Last, our characterization of the independent sets in periodic
hypergraphs holds potential for various applications in operational
research problems~\cite{karypis1997multilevel,alon1999independent} as
well as transportation systems~\cite{harrod2011modeling}. The insights
gained from our research can be leveraged to optimize decision-making
processes in these domains.

\subsection{Paper Organization}

The remainder sections of the paper are organized as follows. In Sec.~\ref{sec:model},
we present the network model and introduce the fundamental properties of the periodic graph induced by the network model. Additionally, we extend the isomorphism and connectivity properties of periodic graphs to periodic hypergraphs.
Sec.~\ref{sec:rr} focuses on basic theoretical results. We provide a proof that the scheduling rate region can be attained through collision-free, periodic schedules, and establish the convexity of the scheduling rate region. Moreover, we explore how the isomorphism and connectivity properties of periodic hypergraphs can simplify the rate region problem.
Moving on to Sec.~\ref{sec:framed}, we characterize the achievable rates
using scheduling with guard intervals.

Our main results are presented from Sec.~\ref{sec:sg} to Sec.~\ref{sec:algm}. 
In Sec.~\ref{sec:sg}, we introduce the concept of scheduling graphs and demonstrate that a collision-free schedule is equivalent to a directed path in a scheduling graph. We explore the use of cycles and paths within a scheduling graph to characterize the scheduling rate region effectively. Additionally, in Sec.~\ref{sec:bin}, we enhance certain results specifically for the binary collision model. 
Moving on to Sec.~\ref{sec:alg}, we investigate the dominance property of step-$T$ scheduling graphs. By analyzing this property, we derive refined characterizations of the rate region and develop algorithms to compute the scheduling rate region. 
In Sec.~\ref{sec:algm}, we study how to maximize a linear function over the scheduling rate region.

Lastly, Sec.~\ref{sec:con} serves as the concluding remarks, where we discuss possible extensions of our results and future research directions.  To facilitate understanding and reference, we have compiled a list of notations used throughout the paper in 
Table~\ref{tab:notation}.

\begin{table}
  \caption{Some notations used in the paper, listed in the alphabetical order.}\label{tab:notation}
  \centering %
  \begin{tabular}{lp{12cm}l}
    \toprule
    Notation & Explanation & Section \\
    \midrule %
    $\conv \mc A$ & the convex hull of a set $\mc A$ & \ref{sec:pssg} \\
    $\cl(P)$ & the cycle generated from a path $P=(A_0,\ldots,A_k)$ in $(\mc M_T, \mc E_T)$ & \ref{sec:dom} \\
    $D_{\mc L}$ & the link-wise delay matrix & \ref{sec:linkmodel} \\
    $D^*_{\mc N}$, $D^*$ &  the character of network $\mc N$ & \ref{sec:sregion}  \\
    $\dom \mc A$ & the collection of all $B\in (\mc R^+)^{m\times n}$ that is dominated by some elements in $\mc A$ & \ref{sec:dom} \\
    $\mc E^*$ & $\mc E^* = M^*_{2T}$ & \ref{sec:sgcal} \\
    $\mc I(l)$ & the collision set of a link $l$ & \ref{sec:dnm} \\
    $\mc I$ & the collision profile, $\mc I = (\mc I(l),l\in \mc L)$ & \ref{sec:dnm} \\
    $\mc L$ & the set of communication links & \ref{sec:dnm}\\
    $\max_{\pgeq}\mc A$ & the set of maximal elements of the
partially ordered set $(\mc A, \pgeq)$ & \ref{sec:alg} \\
    $(\mc M_T, \mc E_{T,Q})$ & the scheduling graph with vertex set $\mc M_T$ and edge set $\mc E_{T,Q}$ & \ref{sec:sgcs}\\
    $(\mc M_T, \mc E_{T})$ & the step-$T$ scheduling graph, $\mc E_T = \mc E_{T,T}$ & \ref{sec:sgcs}\\
    $\mc M^*_L$ & $ \{B: (B,B')\in \mc E^* \text{ for certain } B'\}$  & \ref{sec:sgcal}  \\
    $\mc M^*_R$ & $\{B': (B,B') \in \mc E^* \text{ for certain } B\}$  & \ref{sec:sgcal} \\
    $\mc N$ &  the (link-wise) network model, $\mc N = (\mc L, \mc I, D_{\mc L})$ & \ref{sec:linkmodel}\\
    $\mc N^\infty$ & the  periodic hypergraph induced by $\mc N$ & \ref{sec:linkmodel}\\
    $\mc N_{L,K}^{\text{line}}$ & the uniform line network of $L$ hops with the $K$-hop collision model & \ref{sec:linkmodel} \\
    $\mc N^T$ & the  subgraph of $\mc N^{\infty}$ of $T$ columns & \ref{sec:framed} \\
    $Q$ & the step size of a scheduling graph & \ref{sec:sgcs}\\
    $R_P$ & the rate vector of a closed path in a scheduling graph & \ref{sec:pssg}\\
    $R_S^{\mc N}$, $R_S$ & the rate vector of schedule $S$ for network $\mc N$ & \ref{sec:sregion}  \\
    $\cR^{\mc N}$, $\cR$ &  the (scheduling) rate region of network $\mc N$ & \ref{sec:sregion}  \\
        $\mc R^{(\mc M_T,\mc E_{T,Q})}$ & the convex hull of the rate vectors of all the cycles in $(\mc M_T, \mc E_{T,Q})$ & \ref{sec:pssg} \\
    $\mc R_k$ & the subset of $\mc R^{(\mc M_T,\mc E_{T})}$ generated by the cycles of $(\mc M_T, \mc E_{T})$ up to length $k$  & \ref{sec:dom} \\ 
    $\dR^{\mc N^T}$ & the   convex hull of the rate vectors of all the independent sets of $\mc N^T$ & \ref{sec:framed} \\
    $\mathbb R$, $\mathbb R^+$ & $\mathbb R$ is the set of real numbers, $\mathbb R^+$ is the set of non-negative real numbers & \ref{sec:model} \\
    $S(l,t)$ & the entry of a schedule $S$ indexed by the link $l$ and time $t$ & \ref{sec:model} \\
    $S[T,Q,k]$ & the submatrix of a schedule $S$ of $T$ columns, starting from the $kQ$ column & \ref{sec:sgcs} \\
    $T$ & the blocklength of a scheduling graph & \ref{sec:sgcs}\\
    $(\mc V,\mc F)$ & the reduced scheduling graph &\ref{sec:aep} \\
    $\mathbb{Z}$, $\zz$ & $\mathbb{Z}$ is the set of integers, $\zz$ is set of nonnegative integers & \ref{sec:model}\\
    $\pleq$, $\pgeq$ & the partial order relation defined on $\mathbb{R}$, and can be applied on matrices of the same size component-wisely & \ref{sec:rr} \\
    $\land$ & the minimum function of two real numbers, and can be applied on two matrices of the same size component-wisely & \ref{sec:alg} \\
    $\cv 1$ & a column vector with all entries $1$ & \ref{sec:amp} \\
    \bottomrule
  \end{tabular}
\end{table}

\section{Hypergraph Network Model and Periodic Hypergraph}
\label{sec:model}

We propose a general network model that consists of a matrix that
specifies the delays and a \emph{directed hypergraph} that describes
the collision relations. Since the matrix contains only integer
values, this model is also known as a \emph{discrete network
  model}. We formulate the scheduling problem and demonstrate its
relationship to the periodic hypergraph induced by the network model.

Let $\mathbb{Z}$ denote the set of integers and $\zz$ represent the
set of nonnegative integers. Similarly, let $\mathbb{R}$ denote the
set of real numbers, and $\mathbb{R}^+$ correspond to the set of
non-negative real numbers.

\subsection{Discrete Network Model}
\label{sec:dnm}

We start with a node-based network model, but as we progress, we will discover that utilizing only network links is adequate for solving the network scheduling problem.

Suppose time is slotted and each timeslot is indexed by an integer
$t \in \mathbb{Z}$. Consider a network of $N$ nodes indexed by
$1, 2, \ldots, N$. Each node has the capability to both transmit and
receive a specific communication signal within a timeslot. The signal
transmitted by node $i$ at timeslot $t$ propagates to node $j$ at time
$t + D(i,j)$, where $D(i,j) \in \mathbb{Z}^+$ represents the signal
propagation delay from node $i$ to node $j$.  The transmission of node
$i$ in the timeslot $t$ does not affect the reception of node $j$ in
any other timeslots.  The matrix $D=(D(i,j), 1\leq i, j \leq N)$ is
called the \emph{delay matrix} of the network.

In our network model, a \emph{(communication) link} is represented by an ordered pair $(s, r)$, where $1 \leq s \neq r \leq N$, indicating the transmitting and receiving nodes, respectively. Links are directional, meaning that $(i, j)$ and $(j, i)$ are considered distinct links.
Let $\mathcal{L}$ denote a finite set of all the links. For a given link $l$, we use $\tx_l$ and $\rx_l$ to denote the transmitting node and the receiving node of $l$, respectively. Each link can be in one of two states: active or inactive. A link $l$ is considered \emph{active} in a timeslot $t$ if the transmitting node $\tx_l$ sends a signal in timeslot $t$ intended to be received by node $\rx_l$ in timeslot $t + D(\tx_l, \rx_l)$. Conversely, a link $l$ is deemed \emph{inactive} in a timeslot $t$ if no signal is transmitted by node $\tx_l$ during that timeslot.

\begin{example}[Uniform line networks]\label{exam:line}
  A line network consisting of $L$ hops comprises $L+1$ nodes and the link set defined as:
  \begin{equation*}
    \mc L = \{l_i\triangleq (i,i+1): i = 1,\ldots,L\}. 
  \end{equation*}
  In this network, the delay matrix $D$ is defined such that $D(i, j) = |i - j|$ for $1 \leq i, j \leq L+1$. Throughout this paper, we will utilize this network as an example for various definitions and results.
\end{example}

To incorporate the constraints of link activation, we assign to each link $l$ a subset $\mc I(l)$ of $2^{\mc L}$, called the \emph{collision set} of $l$. 
Each subset of links in the collision set $\mc I(l)$ has the potential to impact the reception of node $\rx_l$. In general, when link $l$ is active in timeslot $t$, we declare that a \emph{collision occurs} if there exists a subset $\theta \in \mc I(l)$ such that each link $l'\in \theta$ is also active in timeslot $t+D(\tx_l,\rx_l)-D(\tx_{l'},\rx_l)$, i.e., the signal transmitted by $\tx_{l'}$ propagates to $\rx_l$ in the timeslot $t+D(\tx_l,\rx_l)$. In other words, when all the links in $\theta$ are active in specific timeslots such that their signals simultaneously propagate to node $\rx_l$ at the same timeslot $t+D(\tx_l,\rx_l)$, the transmission of link $l$ in timeslot $t$ fails due to a collision.

The collision set defined above is flexible and inclusive of various collision scenarios. To specifically model the scenario where two links $l$ and $l'$ with the same transmitting node (i.e., $\tx_l = \tx_{l'}$) cannot be active simultaneously, we can define the collision sets with ${l'} \in \mathcal{I}(l)$ and ${l} \in \mathcal{I}(l')$. 
To model the constraint of half-duplex communication, where a node cannot transmit and receive signals simultaneously, the collision set $\mathcal{I}(l)$ should include all non-empty subsets of $\{l': \tx_{l'} = \rx_l\}$.

\begin{example}[]\label{exam:coll}
  For the line network in Example~\ref{exam:line} with $L=4$, the collision sets of the links can be defined as follows:
  \begin{IEEEeqnarray*}{rClrCl}
    \mc I(l_1) & = & \{\{l_2\}\}, &
    \mc I(l_2) & = & \{\{l_3\}, \{l_1,l_4\}\},\\
    \mc I(l_3) & = & \{\{\l_4\}\},\quad &
    \mc I(l_4) & = & \emptyset.
  \end{IEEEeqnarray*}
  From these collision sets, we can observe that $\{l_{i+1}\} \in \mathcal{I}(l_i)$ for $i=1,2,3$. This implies that nodes $2, 3, 4$ cannot transmit and receive signals simultaneously. Additionally, the set $\{l_1,l_4\}$ is in the collision set of link $l_2$. This means that if link $l_1$ is active in timeslot $t-1$ and link $l_4$ is active in timeslot $t$, a collision will occur if link $l_2$ is active in timeslot $t$.
\end{example}

In the network model described above, we define
$\mc I = (\mc I(l),l\in \mc L)$ as the \emph{collision profile}.  The
collision relation among links can be represented as a \emph{directed
  hypergraph}, denoted as $(\mc L, \mc I)$, with the vertex set
$\mc L$ and the directed edge set
$\{(l,\theta): l\in \mc L, \theta\in \mc I(l)\}$. It is worth noting
that in a general directed hypergraph with the vertex set $\mc L$, an
edge belongs to $2^{\mc L}\times 2^{\mc L}$
\cite{gallo1993directed}. However, our directed hypergraph
$(\mc L, \mc I)$ is a special case where the tail of an edge must be a
singleton. Therefore, we represent an edge in the hypergraph of our
network model as $(l,\theta)\in \mc L \times 2^{\mc L}$.  The relation
of the hypergraph model and the physical model has been discussed
in~\cite{grokop2011interference}. For the sake of completeness, we
provide details in the Appendix on how to transform the results
obtained to a wireless network in the physical model.

When all the collision sets $\mc I(l)$ consist only of singletons
(i.e., for any $\theta \in \mc I(l)$, $|\theta|=1$), the collision
model and the network model are said to be \emph{binary}. For a binary
collision model, we can represent $\mc I(l)$ as a subset of $\mc L$ to
simplify the notation, and $(\mc L, \mc I)$ becomes a directed graph.
For scheduling with propagation delays, the network model studied in
\cite{grokop2011interference,chitre2012throughput} is a binary model
with $\mc I(l) = \mc L\setminus \{l\}$. While our primary focus is on
the general hypergraph model, we will demonstrate that some of our
results can be further improved for the binary collision model.

\begin{example}[Line network with the $K$-hop collision model] \label{ex:1}
For the line network defined in Example~\ref{exam:line}, we consider a \emph{binary collision model} called the $K$-hop model, where the reception of a node can only have collisions from nodes within $K$ hops distance~\cite{sharma2006complexity}.
For each link $l_i$ with $i = 1,\ldots,L$, the collision set $\mc I(l_i)$ of the $K$-hop model is defined as:
\begin{equation}\label{eq:exi}
\mc I(l_i) = \{l_j: j\neq i, |j-i-1| \leq K\}.
\end{equation}
When $K\geq 1$, it can be observed that for $i=1,\ldots,L-1$, $l_{i+1} \in \mc I(l_i)$. This implies that node $i+1$ is half-duplex, meaning it cannot transmit and receive signals simultaneously. In the case of $L=4$ and $K=1$, the collision sets are as follows:
  \begin{equation}
    \label{eq:k1}
    \begin{IEEEeqnarraybox*}{rCl.rCl} 
      \mc I(l_1) & = & \{l_2,l_3 \},\quad &
      \mc I(l_2) & = & \{l_3,l_4\}, \\
      \mc I(l_3) & = & \{l_4\},  &
      \mc I(l_4) & = & \emptyset.
    \end{IEEEeqnarraybox*}
  \end{equation}
  Note that the collision relation among links is not necessarily symmetric. In the example where $L=4$ and $K=1$, link $l_3$ may generate collisions for $l_1$, but $l_1$ does not generate collisions for $l_3$. This can be understood by examining the corresponding network nodes: link $l_1$ represents the communication from node $1$ to node $2$, while $l_3$ represents the communication from node $3$ to node $4$. In the $K$-hop model with $K=1$, the transmission of node $3$ can affect the reception of node $2$, but the transmission of node $1$ cannot affect the reception of node $4$. 

\end{example}

When all delays are set to $0$, the network model we defined corresponds to a model without considering delays. In the literature, undirected hypergraphs (or graphs) are commonly used to model collisions in such scenarios~\cite{arikan1984some,hajek1988link,ephremides1990scheduling,tassiulas92stability,lin2006,wu2007scheduling,chaporkar2008throughput,
  mceliece1994performance,sarkar1998hypergraph,li2012maximal,zhang2016radio,ganesan2021some}.
The reason for the collision model being undirectional is that the link scheduling does not depend on the direction in this special case.
However, when considering general delays, it becomes necessary to use a \emph{directed} hypergraph (or graph) to accurately model collisions. The necessity of using a directed hypergraph will be further elaborated after the link scheduling problem is formulated.

\subsection{Link-wise Network Model and Link Schedule}
\label{sec:linkmodel}

To simplify the network model, we define the $|\mc L|\times |\mc L|$ \emph{link-wise delay matrix}  $D_{\mc L}$ with
\begin{equation*}
  D_{\mc L}(l,l') = D(\tx_l,\rx_l)-D(\tx_{l'},\rx_l).
\end{equation*}
The definition of the link-wise delay matrix does not depend on the collision profile.
Collision can be determined using $D_{\mc L}$. Specifically, if a link $l$ is active in a timeslot $t$, it has a collision if for a certain $\theta \in \mc I(l)$, every link $l'\in \theta$ is also active in the timeslot $t+D_{\mc L}(l,l')$. By utilizing the link-wise delay matrix $D_{\mc L}$, it is not necessary to directly refer to the network nodes when verifying collisions.

We define the (link-based) network model as $\mathcal{N}\triangleq (\mc L, \mc I, D_{\mc L})$. It is sufficient for us to use this link-based model in the following discussion. In the network model, $(\mc L, \mc I)$ represents a directed hypergraph of finite size. The
entries of $D_{\mc L}$ are integers and can be negative. 
If $l'\in \cup_{\theta\in \mc I(l)} \theta$, then the value $D_{\mc L}(l,l')$ is required in collision checking and we consider the $(l,l')$ entry of $D_{\mc L}$ as \emph{relevant}.
On the other hand, if the $(l,l')$ entry of $D_{\mc L}$ is \emph{not relevant}, meaning $l'\notin \cup_{\theta\in \mc I(l)} \theta$, then $D_{\mc L}(l,l')$ is not involved in collision checking. When the  $(l,l')$ entry of $D_{\mc L}$ is not relevant, we mark this entry as $*$ in a link-wise delay matrix.
When the context is clear, we also call $D_{\mc L}$ the delay matrix.

\begin{example}\label{ex:2} %
  Following Example~\ref{ex:1}, the link-wise delay matrix $D_{\mc L}$ of the $L$-length, $K$-hop collision line network is given by:
  \begin{equation}\label{eq:exd}
    D_{\mc L}(l_i,l_j) = D(i,i+1) - D(j,i+1) = 1 - |j-i-1|.
  \end{equation}
  The network is denoted as $\mc N_{L,K}^{\text{line}}=(\mc L, \mc I, D_{\mc L})$, where $\mc L=\{l_1,\ldots,l_L\}$, $\mc I$ is defined in \eqref{eq:exi}, and $D_{\mc L}$ is defined in \eqref{eq:exd}.
  The graphical representation of $\mc N_{4,1}^{\text{line}}$ and $\mc N_{4,2}^{\text{line}}$ are
  shown in Fig.~\ref{fig:line41}. %
\end{example}

\begin{figure}[tb]
  \centering
  \subfigure[$\mc N_{4,1}^{\text{line}}$]{
  \begin{tikzpicture}[link/.style={circle, draw, thin,fill=cyan!20}]
    \foreach \place/\x in {{(-3,0)/1},{(-1,0)/2}, {(1,0)/3},{(3,0)/4}}
    \node[link,label=below:$l_\x$] (u\x) at \place {};
    
    \draw[->] (u1) -- (u2) node[midway,auto] {$1$};
    \draw[->] (u2) -- (u3) node[midway,auto] {$1$};
    \draw[->] (u3) -- (u4) node[midway,auto] {$1$};
    \draw[->, bend left=45] (u1) to node[above] {$0$} (u3);
    \draw[->, bend left=45] (u2) to node[above] {$0$} (u4);
  \end{tikzpicture}
  }
  \subfigure[$\mc N_{4,2}^{\text{line}}$]{
    \begin{tikzpicture}[link/.style={circle, draw, thin,fill=cyan!20}]
    \foreach \place/\x in {{(-3,0)/1},{(-1,0)/2}, {(1,0)/3},{(3,0)/4}}
    \node[link,label=below:$l_\x$] (u\x) at \place {};
    
    \draw[->, bend left=30] (u1) to node[above] {$1$} (u2);
    \draw[->, bend left=30] (u2) to node[above] {$-1$} (u1);
    \draw[->, bend left=30] (u2) to node[above] {$1$} (u3);
    \draw[->, bend left=30] (u3) to node[above] {$-1$} (u2);
    \draw[->, bend left=30] (u3) to node[above] {$1$} (u4);
    \draw[->, bend left=30] (u4) to node[above] {$-1$} (u3);    
    
    \draw[->, bend left=50] (u1) to node[above] {$0$} (u3);
    \draw[->, bend left=50] (u2) to node[above] {$0$} (u4);
    \draw[->, bend right=35] (u1) to node[above] {$-1$} (u4);
  \end{tikzpicture}
  }
  \caption{The graphical representation of $\mc N_{4,1}^{\text{line}}$ and $\mc N_{4,2}^{\text{line}}$. In these graphs, as well as the following graphical representation of our discrete network models, the vertices in a graph represent links in the network. The number on an edge $(l,l')$ is the value of $D_{\mc L}(l,l')$.} \label{fig:line41}
\end{figure}

One fundamental question related to a discrete network 
$\mc N =(\mathcal{L},\mathcal{I},D_{\mc L})$ is the efficiency of link activation scheduling. 
A \emph{(link) schedule} $S$ is a matrix of binary digits indexed by pairs $(l,t)\in \mathcal{L}\times \mathbb{Z}$, where
$S(l,t)=1$ indicates that $l$ is active in timeslot $t$, and $S(l,t)=0$ indicates that link $l$ is inactive in timeslot $t$. 

\begin{definition}[Collision-free schedule]\label{def:cf}
For a given schedule $S$ and a pair $(l,t)\in \mathcal{L}\times \mathbb{Z}$, we say that $S(l,t)$ has a \emph{collision} in the network $\mc N=(\mathcal{L},\mathcal{I},D_{\mc L})$ if for a certain $\theta \in  \mc I(l)$, we have $S(l',t+D_{\mc L}(l,l')) = 1$ for every $l'\in\theta$. On the other hand, if for all $\theta \in  \mc I(l)$,  $S(l',t+D_{\mc L}(l,l')) = 0$ for a certain $l'\in \theta$, we say $S(l,t)$ is \emph{collision-free}.
A schedule $S$ is said to be \emph{collision-free} if $S(l,t)$ is collision-free for all $(l,t) \in \mc L\times \mathbb Z$ with $S(l,t)=1$.  
\end{definition}

\begin{definition}[Periodic hypergraph]
  Consider a network $\mc N=(\mathcal{L},\mathcal{I},D_{\mc L})$. The \emph{periodic (undirected) hypergraph} induced by $\mc N$, denoted by $\mc N^{\infty}$, has the vertex set $\mathcal{L}\times \mathbb{Z}$. In $\mc N^{\infty}$,
  a subset $\{(l_i,t_i):i=1,\ldots,k\} \subset \mathcal{L}\times \mathbb{Z}$ is an edge if and only if there exists $j \in \{1,\ldots,k\}$ such that $\{l_i:i \in\{1,\ldots,k\}, i\neq j\} \in \mc I(l_j)$ and $t_i = t_{j} + D_{\mc L}(l_{j},l_i)$ for all $i\neq j \in \{1,\ldots,k\}$. In other words,
  an edge is always of the form $\{(l,t),(l',t+D_{\mc L}(l,l')):l'\in \theta \}$ for some $\theta \in \mc I(l)$. 
\end{definition}

When all the collision sets are binary, $\mc N^{\infty}$ becomes a graph with edges of the form $\{(l,t),(l',t+D(l,l'))\}$ for all $l'\in \mc I(l)$. See Fig.~\ref{fig:periodic}-(a) for an illustration of the periodic graph induced by $\mc N_{4,1}^{\text{line}}$.
For a general hypergraph with the vertex set $\mc V$ and edge set $\mc E\subset 2^{\mc V}$, a subset $\mc A$ of $\mc V$ is said to be \emph{independent} if for any $\mc U\in \mc E$, $\mc U \nsubseteq \mc A$.
The following theorem establishes the relation between a collision-free schedule of $\mc N$ and an independent set of $\mc N^{\infty}$.

\begin{theorem}\label{thm:ip}
  A schedule $S$ is collision-free for a network $\mc N=(\mathcal{L},\mathcal{I},D_{\mc L})$ if and only if the set $\{(l,t)\in \mathcal{L}\times \mathbb{Z}:S(l,t)=1\}$ is an {independent set} in $\mc N^{\infty}$.
\end{theorem}
\begin{IEEEproof}
  For a schedule $S$, let $\mc A = \{(l,t)\in \mathcal{L}\times \mathbb{Z}:S(l,t)=1\}$.
  For any edge $e$ of $\mc N^{\infty}$, $e =\{(l,t),(l',t+D_{\mc L}(l,l')):l'\in \theta \}$ for some $\theta \in \mc I(l)$. When $S$ is collision-free, either $S(l,t)=0$, or $S(l,t)=1$ and there exists $l'\in \theta$ such that $S(l',t+D_{\mc L}(l,l')) = 0$. Therefore, $e$ is not a subset of $\mc A$ and hence $\mc A$ is an independent set of $\mc N^{\infty}$.

  Suppose $S$ is not collision-free. Then there exists $(l,t) \in \mc L\times \mathbb Z$ with $S(l,t)=1$, and a certain $\theta\in \mc I(l)$ such that $S(l',t+D_{\mc L}(l,l')) = 1$ for all $l'\in \theta$. We observe in this case that $e=\{(l,t),(l',t+D_{\mc L}(l,l')):l'\in \theta \}$ is an edge of $\mc N^{\infty}$ and $e\subset \mc A$. 
Therefore, $\mc A$ is not an independent set in $\mathcal{N}^{\infty}$.
\end{IEEEproof}

\begin{figure}[t]
	\centering
  \subfigure[]{
    \begin{tikzpicture}[link/.style={circle, draw, thin,fill=cyan!20},scale=0.45]

      \foreach \x in {1,2,3,4}
      {
        \foreach \y in {0,1,2,3,4,5}
        {
          \node[link] (a\x\y) at (2.3*\y,-1.8*\x) {};
        }
        \node[left=1mm of a\x0] {$l_{\x}$};
        \node[right=2mm of a\x5] {$\ldots$};
      }

      \foreach \y in {0,1,2,3,4,5}
      {
        \draw[ bend left=45] (a1\y) to (a3\y);
        \draw[ bend left=45] (a2\y) to (a4\y);
      }

      \foreach \y/\z in {0/1,1/2,2/3,3/4,4/5}
      {
        \draw[] (a1\y) -- (a2\z);
        \draw[] (a2\y) -- (a3\z);
        \draw[] (a3\y) -- (a4\z);
      }
    \end{tikzpicture}
  } \quad
  \subfigure[]{
    \begin{tikzpicture}[link/.style={circle, draw, thin,fill=cyan!20},scale=0.45]

      \foreach \y in {0,1,2,3,4,5}
      {
        \foreach \x in {1,2,3,4}
        {
          \node[link] (a\x\y) at (2.3*\y,-1.8*\x) {};
        }
        \draw[ bend left=45] (a1\y) to (a3\y);
        \draw[ bend left=45] (a2\y) to (a4\y);
        \draw[] (a1\y) -- (a2\y);
        \draw[] (a2\y) -- (a3\y);
        \draw[] (a3\y) -- (a4\y);
      }

      \foreach \x in {1,2,3,4}
      {
        \node[left=1mm of a\x0] {$l_{\x}$};
        \node[right=2mm of a\x5] {$\ldots$};
      }
    \end{tikzpicture}
    }
    \caption{Illustration of the periodic graphs. (a) is the periodic graph induced by $\mc N_{4,1}^{\text{line}}$. (b) is the periodic graph induced by a network that shares the same link set and collision profile as $\mc N_{4,1}^{\text{line}}$, but has $\cv 0$ as the delay matrix.}
  \label{fig:periodic}
\end{figure}

Theorem~\ref{thm:ip} gives an equivalence relation between a collision-free schedule $S$ of $\mc N$ and an independent set $I$ of $\mc N^\infty$. Specifically, the support of $S$ forms an independent set in $\mc N^\infty$, and the indicator function of $I$, represented as a binary matrix, serves as a collision-free schedule for $\mc N$.
Based on this equivalence relation, we can use one representation to refer to the other interchangeably.

Denote by $\cv 0$ the matrix with all the entries
$0$. The network $(\mc L, \mc I, \cv 0)$ has a special periodic graph
where for each $t\in \mathbb{Z}$, the set $\{(l,t),l\in \mc L\}$ forms a
component that is isomorphic to $(\mc L, \mc I)$, with the edge directions ignored. 
Fig.~\ref{fig:periodic}-(b) illustrates the periodic graph of the
network generated by replacing the delay matrix in
$\mc N_{4,1}^{\text{line}}$ as $\cv 0$. The independent sets of the
periodic graph of $(\mc L, \mc I, \cv 0)$ can be completely
characterized by the independent sets of
$(\mc L, \mc I)$, where the edge directions are ignored. This is the reason why the scheduling problem without considering delays does not require the edge directions in
$(\mc L, \mc I)$.

The equivalence between a collision-free schedule of $\mc N$ and an
independent set of $\mc N^{\infty}$ cannot provide an explicit and
exactly solution to the scheduling problem in general when
$D_{\mc L} \neq \cv 0$. This is because the periodic graph has
\emph{infinitely} many vertices, which means that not only can an
independent set have an unbounded size, but also the number of
independent sets is infinite. In Sec.~\ref{sec:framed}, we will
discuss how existing approaches to independent sets can only
approximate the optimal scheduling solutions, and the corresponding
computation cost is high. While some properties of periodic graphs, such as isomorphism and connectivity, have been studied in the literature~\cite{orlin1984some, beckenbach2019matchings}, the independent set problem has not been well understood.
In Sec.~\ref{sec:rr}, we will formally define the scheduling rate region problem.
In Sec.~\ref{sec:sg}, we will further investigate the properties of the
periodic hypergraph to enable an exact and explicit solution of the scheduling problem.

\subsection{Useful Properties of Periodic Hypergraphs}
\label{sec:decom}

Here, we will briefly introduce the isomorphism and connectivity properties
of periodic graphs and discuss their extension to periodic
hypergraphs. In Sec.~\ref{sec:sim}, we will leverage these properties to
simplify the scheduling rate region problem.

\subsubsection{Isomorphism}
A \emph{vertex assignment} for a network $\mc N = (\mc L, \mc I, D_{\mc L})$ is an integer-valued vector $\mathbf{b}=(b_l,l\in \mc L)$. Each vertex assignment $\mathbf{b}$ induces a new link-wise delay matrix $D_{\mc L}^{\mathbf{b}} =  (D_{\mc L}^{\mathbf{b}} (l,l'))$ where
\begin{equation}\label{eq:db}
D_{\mc L}^{\mathbf{b}}(l,l') =  D_{\mc L}(l,l')+b_{l}-b_{l'},
\end{equation}
and hence a new network $\mc N_{\mathbf{b}} = (\mc L, \mc I, D_{\mc L}^{\mathbf{b}})$. According to \cite{orlin1984some}, if $(\mc L, \mc I)$ is a graph, $\mc N^{\infty}$ and $\mc N_{\mathbf{b}}^{\infty}$ are isomorphic with respect to the bijection $f: \mc L\times \mathbb Z \rightarrow \mc L\times \mathbb Z$ with $f(l,t) = (l,t+b_l)$. In other words, $\mc N_{\mathbf{b}}^{\infty}$ is obtained by shifting all the vertices in the row $l$ of $\mc N^{\infty}$ by $b_l$.
The mapping is still an isomorphism when $(\mc L, \mc I)$ is a hypergraph as the argument in \cite{orlin1984some} involves only the delay matrix.

\subsubsection{Connectivity}

In an undirected graph, two vertices are said to be connected if there exists a path between these two vertices. 
Exploring the connectivity of $\mc N^{\infty}$ can potentially simplify the  scheduling problem by considering each component of $\mc N^{\infty}$ individually. We first discuss the connectivity when $(\mc L, \mc I)$ is a graph, which has been studied in \cite{orlin1984some}.
Let $g_{\mc N}$ be the greatest common divisor of $D_{\mc L}(l,l')$ for all $l\in \mc L$ and $l'\in \mc I(l)$. Under the condition that $D_{\mc L}\neq \cv 0$,  $g_{\mc N}$ is well-defined. 
Then $D_{\mc L}/g_{\mc N}$ is an integer matrix. According to \cite{orlin1984some}, $\mc N^{\infty}$ has $g_{\mc N}$ components isomorphic to the periodic graph of $(\mc L, \mc I, D_{\mc L}/g_{\mc N})$. 
Fig.~\ref{fig:connectivity_eg} illustrates a periodic graph with three components. 
In contrast to the case $D_{\mc L} = \cv 0$, the components of $\mc N^{\infty}$  in general have an infinite size. 

\begin{figure}[t]
	\centering

    \begin{tikzpicture}[link/.style={circle, draw, thin,fill=cyan!20},scale=0.5]

      \foreach \x in {1,2,3,4}
      {
        \foreach \y in {0,3}
        {
          \node[link] (a\x\y) at (2.3*\y,-1.8*\x) {};
        }
        \foreach \y in {1,4}
        {
          \node[link,fill=red!35] (a\x\y) at (2.3*\y,-1.8*\x) {};
        }
        \foreach \y in {2,5}
        {
          \node[link,fill=gray!90] (a\x\y) at (2.3*\y,-1.8*\x) {};
        }
        \node[left=1mm of a\x0] {$l_{\x}$};
        \node[right=2mm of a\x5] {$\ldots$};
      }

      \foreach \y in {0,1,2,3,4,5}
      {
        \draw[] (a1\y) to (a2\y);
        \draw[] (a2\y) to (a3\y);
        \draw[] (a3\y) to (a4\y);
      }

      \foreach \y/\z in {0/3,1/4,2/5}
      {
        \draw[] (a1\y) -- (a3\z);
        \draw[] (a2\y) -- (a4\z);
      }
    \end{tikzpicture}
    \caption{An periodic graph with three components. Each component is illustrated by a different color (gray scale).}
  \label{fig:connectivity_eg}
\end{figure}

If $(\mc L, \mc I)$ is a hypergraph, for $l\in \mc L$, define 
$	\mc I'(l) = \bigcup_{\theta\in \mc I(l)} \theta$. 
	Let $\mc I' = (\mc I'(l), l\in \mc L)$. Then $\mc N' = (\mc L, \mc I', D_{\mc L})$ is a new network with a binary collision model. By~\cite{beckenbach2019matchings}, two vertices in $\mc N^\infty$ are connected if and only if the two corresponding vertices in $(\mc N')^\infty$ are connected. Let $g_{\mc N'}$ be the greatest common divisor of $D_{\mc L}(l,l')$ for all $l\in \mc L$ and $l'\in \mc I'(l)$. We know that $\mc N^{\infty}$ has $g_{\mc N}\triangleq g_{\mc N'}$ components isomorphic to the periodic graph of $(\mc L, \mc I, D_{\mc L}/g_{\mc N})$.

\section{Scheduling Rate Region}
\label{sec:rr}

When delays are all $0$, an independent set of $(\mc L,\mc I)$ with the edge directions ignored represents an achievable scheduling rate vector, and the set of all maximal independent sets determines the scheduling rate region~\cite{ephremides1990scheduling, tassiulas92stability,sharma2006complexity}. However, for networks with general delays, the concepts of achievable scheduling rate vectors and the scheduling rate region need to be extended to account for the characteristics of general periodic hypergraphs.
In this section, we will formally define the schedule rate vector and the scheduling rate region for a general network with delays. We will discuss some fundamental properties of the scheduling rate region, and study a class of special schedules known as guarded schedules. 

For two real matrices $A$ and $B$ of the same
size, we write $A \pleq B$ if all the entries of $A$ are not larger
than the corresponding entries of $B$ at the same position. We
similarly define $A \pgeq B$ to indicate that all entries of $A$ are not smaller than the corresponding entries of $B$ at the same position. For a matrix $A$ and a scalar $a$, we
write $A+a$ to denote the matrix obtained by adding $a$ to each entry of $A$. 
We similarly define $A-a$ to be the matrix obtained by subtracting $a$ from each entry of $A$.

\subsection{Scheduling Rate Vector}
\label{sec:sregion}

For a network $\mc N = (\mc L,\mc I,D_{\mc L})$, we denote for each schedule $S$ and  link $l$
\begin{equation}\label{eq:sr}
R_{S}^{\mc N}(l) = \lim_{T\rightarrow \infty} \frac{1}{T} \sum_{t=0}^{T-1} \iota(S(l,t)=1, S(l,t) \text{ is collision-free}),
\end{equation}
where $\iota(A_1,A_2,\ldots)$ is the indicator function with a value $1$ if the sequence of conditions $A_i$ are all true, and $0$ otherwise.
To maintain consistency with network scheduling conventions, we only consider $S(l,t)$ with $t\geq 0$ when defining $R_S^{\mc N}(l)$. 
If the limit on the right-hand side of \eqref{eq:sr} exists, we say that $R_{S}^{\mc N}(l)$ exists.
When $R_{S}^{\mc N}(l)$ exists, we call $R_{S}^{\mc N}(l)$ the \emph{(scheduling) rate of link $l$}.
If $R_S^{\mc N}(l)$ exists for all $l\in \mc L$, 
we call $R_S^{\mc N} = (R_S^{\mc N}(l),l\in \mc L)$ the \emph{rate vector} of $S$ for $\mc N$. We may omit the superscript in $R_S^{\mc N}$ and $R_S^{\mc N}(l)$ when the network $\mc N$ is implied.

\begin{definition}[Scheduling rate region]\label{def:1}
  For a network $\mc N = (\mc L,\mc I,D_{\mc L})$,
  a rate vector $R=(R(l),l\in \mc L)$ is said to be \emph{achievable} if for any $\epsilon>0$, there exists a schedule $S$ such that $R_S\pgeq R-\epsilon$. The set $\cR^{\mc N}$ of all the achievable rate vectors is called the \emph{scheduling rate region} of $\mathcal{N}$.
\end{definition}

Define the \emph{character} of the network $\mc N$ as
\begin{equation}\label{eq:dstar}
D^*_{\mc N} = \max_{l\in \mc L}\max_{\theta\in \mc I(l)} \max_{l'\in \theta}|D_{\mc L}(l,l')|.
\end{equation}
In other words, $D^*_{\mc N}$ is 
the maximum relevant delay in $D_{\mc L}$.
When $\mc N$ is known from the context, we also write $D^*_{\mc N}$ as $D^*$.

\begin{example}\label{ex:3}%
  For $\mc N_{L,K}^{\text{line}}$ defined in Example~\ref{ex:2}, by \eqref{eq:exi} and \eqref{eq:exd}, 
  \begin{IEEEeqnarray*}{rCl}
    D^* %
    & = &  \max_{1\leq i\neq j \leq L, |j-i-1| \leq K} |1 - |j-i-1|| \\
    & = & \max\{\min\{L,K\}-1,1\}.
  \end{IEEEeqnarray*}
  So, when $K=1$, $D^* = 1$, and when $L\geq K\geq 2$, $D^* = K-1$.
\end{example}

\begin{definition}[Periodic schedule]
A schedule $S$ is considered \emph{periodic} if there exists a positive integer $T_p$ such that $S(l,t) = S(l,t+T_p)$ for all $(l,t) \in \mc L\times \mathbb Z$. The positive integer $T_p$ in this context is called a \emph{period} of the schedule. 
\end{definition}

Similar to Definition~\ref{def:1}, a rate vector $R$ is considered
\emph{achievable by collision-free, periodic schedules} if for any $\epsilon >0$, there exists a collision-free, periodic schedule $S$ such that
$R_S\pgeq R-\epsilon$.
Although the scheduling rate region $\cR^{\mc N}$ is defined for general schedules, 
the following lemma states that collision-free, periodic schedules achieve the rate region $\cR^{\mc N}$. 
Our result directly implies the special cases observed in the existing papers \cite{grokop2011interference,chitre2012throughput} when the network has a binary collision with $\mc I(l) = \mc L\setminus \{l\}$.
Note that we do not limit the period of the periodic schedules in Lemma~\ref{prop:cfr}. In the next section, we will further enhance the result by showing that we only need a finite set of periodic schedules to completely characterize $\cR^{\mc N}$.

\begin{lemma}\label{prop:cfr}
  For a network $\mc N$, the  rate  region $\cR^{\mc N}$  can be achieved using only collision-free, periodic schedules.
\end{lemma}

\begin{remark}
  Our proof is based on a constructive approach. First, for any given schedule, we can always find a collision-free schedule with the same rate vector by setting all the entries corresponding to collisions to $0$. Second, for any collision-free schedule with a rate vector $R$, we can construct a periodic schedule using a segment of the schedule from time $0$ to $T-1$, such that the rate vector of the periodic schedule converges to $R$ as $T$ tends to infinity.
\end{remark}

\begin{IEEEproof}%
	Fix $R\in \cR$ and $\epsilon>0$. By Definition~\ref{def:1}, there exists a schedule $S$ such that
	\begin{equation}
	\label{eq:pd1}
	R_{S}(l)\geq R(l)-\epsilon/2, \text{ for all } l\in \mc L.
	\end{equation}
	 Define a schedule $S'$ such that
	\begin{equation*}
	S'(l,t) =
	\begin{cases}
	1 & \text{$S(l,t)=1$ and is collision-free} \\
	0 & \text{otherwise}.
	\end{cases}
	\end{equation*}
	We see that $S'$ is collision-free and $R_{S'} = R_S$. 
	
	By the definition in \eqref{eq:sr}, there exists a sufficiently large $T_0$ such that for all $T\geq T_0$ and all $l\in \mc L$,
	\begin{equation}\label{eq:xid}
	\left |R_{S'}(l) - \frac{1}{T}\sum_{t=0}^{T-1} S'(l,t) \right| \leq \frac{\epsilon}{4}.
	\end{equation}
	Fix any $T^*\geq \max\{T_0+ D^*, 4D^*/\epsilon\}$. Define a schedule $S^*$ with period $T^*$: %
	\begin{equation*}
	S^*(l,t) =
	\begin{cases}
	S'(l,t) & t = 0,1,\ldots, T^*-1-D^{*},\\
	0 & t = T^*-D^*,\ldots,T^*-1.
	\end{cases}
	\end{equation*}
	Now we argue that $S^*$ is collision-free.
	
	Fix $(l,t)$ with $S^*(l,t)=1$. According to the definition of $S^*$, there exists $t_0 \in \{0,1,\ldots,T^*-1-D^*\}$ such that $t=kT^* + t_0$. We show that $S^*(l,t)$ is collision-free by contradiction. Assume there exists $\theta\in \mc I(l)$ such that $S^*(l',t+D_{\mc L}(l,l')) = 1$ for every $l'\in\theta$, i.e., $S^*(l,t)$ has a collision. For $l'\in \theta$, let $t' = t+D_{\mc L}(l,l')$. 
	As $l'\in \theta \in \mc I(l)$, we have $|D_{\mc L}(l,l')| \leq D^*$, and hence $kT^*-D^*\leq t'\leq kT^*+T^*-1$. We discuss the possible range of $t'$ in three cases: 
	\begin{enumerate}
        \item When $kT^*-D^* \leq t' < kT^*$, by the definition of $S^*$, $S^*(l',t') = 0$.
        \item When $(k+1)T^*-D^* \leq t'\leq (k+1)T^*-1$, by the definition of $S^*$, $S^*(l',t') = 0$.
        \item When $kT^*\leq t' \leq (k+1)T^*-1-D^*$, write $t' = kT^*+ t_0'$.
	Due to the periodical property of $S^*$, 
	$S^*(l',t') = S^*(l',t_0')$.
	As $t_0' \in \{0,1,\ldots,T^*-1-D^*\}$, by the definition of $S^*$,
	$S^*(l',t_0') = S'(l',t_0')$.
	Similarly, we have $S'(l,t_0) = S^*(l,t_0)= S^*(l,t) = 1$. As $S'(l,t_0)$ is collision-free, for certain $l'\in \theta$, $S'(l',t_0') = 0$, i.e., $S^*(l',t')=0$.
      \end{enumerate}
      Therefore, for all the three cases of $t'$, we get a contradiction to the assumption that $S^*(l,t)$ has a collision.

	As $S^*$ is periodic and collision-free, we further have
	\begin{IEEEeqnarray*}{rCl}
		R_{S^*}(l) & = & \frac{1}{T^*}\sum_{t=0}^{T^*-1}S^*(l,t) \\
		& = & \left(1- \frac{D^*}{T^*}\right) \frac{1}{T^*-D^*} \sum_{t=0}^{T^*-1-D^*}S'(l,t) \\
		& \geq & \left(1- {D^*}/{T^*}\right)(R_{S'}(l) -\epsilon/4) \\
		& \geq & R_{S'}(l) -\epsilon/4 - D^*/T^* \\
		& \geq & R_{S'}(l) - \epsilon/2 \\
		& \geq & R(l) - \epsilon,
	\end{IEEEeqnarray*}
	where the first inequality follows from $T^*\geq T_0+ D^*$ and \eqref{eq:xid}, the third inequality follows from $T^* \geq 4D^*/\epsilon$, and the last inequality is obtained by substituting \eqref{eq:pd1}.
	The proof of the theorem is complete.
  \end{IEEEproof}

The convexity is another fundamental property of the scheduling rate region $\cR^{\mc N}$. In the next section, we will further show that $\cR^{\mc N}$ is a polytope with a finite vertices. 

\begin{lemma}\label{lemma:convex}
  The rate region $\cR^{\mc N}$ of a network $\mc N$ is convex.
\end{lemma}

\begin{remark}
  Our proof is based on a constructive approach. For any two collision-free, periodic schedules, we can construct a new periodic schedule that has a rate vector close to a convex combination of the rate vectors of the two original schedules.
\end{remark}

\begin{IEEEproof}%
	Fix $R_1$ and $R_2$ in $\cR^{\mc N}$. Let $R = \alpha R_1 + (1-\alpha)R_2$ where $0< \alpha < 1$. The lemma is proved by showing $R \in \cR^{\mc N}$.
	Fix $\epsilon >0$. By Lemma~\ref{prop:cfr}, there exists a collision-free schedule $S_1$ of period $T_1$ such that $R_{S_1}\pgeq R_1-\epsilon/2$, and a collision-free schedule $S_2$ of period $T_2$ such that $R_{S_2}\pgeq R_2-\epsilon/2$.
	
	For a positive integer $k_1$, let $k_2 = \lceil \frac{1-\alpha}{\alpha}\frac{T_1}{T_2}k_1 \rceil$. Construct a schedule $S$ of period $k_1T_1+k_2T_2+2D^*$ such that
	$S(l,t) = S_1(l,t)$ for $t \in \{0,1,\ldots, k_1T_1-1\}$, $S(l,t) = S_2(l,t-k_1T_1-D^*)$ for  $t \in k_1T_1+D^*+\{0,1, \ldots, k_2T_2-1\}$, and $S(l,t) = 0$ for other values of $t$ in the first period. 
	Similar to the proof of Lemma~\ref{prop:cfr}, we can argue that the schedule $S$ is collision-free. The rate vector $R_S$ satisfies
	\begin{IEEEeqnarray*}{rCl}
		R_S & = & \frac{k_1T_1R_{S_1} + k_2T_2R_{S_2}}{k_1T_1+k_2T_2+2D^*} \\
		& \pgeq & \frac{k_1T_1R_{S_1} + \frac{1-\alpha}{\alpha}T_1k_1R_{S_2}}{k_1T_1+ \frac{1-\alpha}{\alpha}T_1k_1 + T_2+2D^*} \\
		& = & \frac{\alpha R_{S_1}+(1-\alpha)R_{S_2}}{1+\alpha(T_2+2D^*)/(T_1k_1)} \\
		& \pgeq & \frac{R-\epsilon/2}{1+\alpha(T_2+2D^*)/(T_1k_1)} \\
		& = & R - \frac{R\alpha(T_2+2D^*)/(T_1k_1)+\epsilon/2}{1+\alpha(T_2+2D^*)/(T_1k_1)}.
	\end{IEEEeqnarray*}
	Therefore, when $k_1$ is sufficiently large, $R_S \pgeq R-\epsilon$, and hence $R\in \cR^{\mc N}$. 
  \end{IEEEproof}

\subsection{Simplification by Isomorphism and Connectivity}
\label{sec:sim}

In Section~\ref{sec:decom}, we have discussed the concepts of isomorphism and connectivity of periodic hypergraphs. Now, we will demonstrate how these properties can be utilized to simplify the problem of scheduling rate region. Our discussion is self-contained, as we solely rely on the properties of schedules and rate regions introduced earlier.

Consider a network $\mc N=(\mc L,\mc I, D_{\mc L})$ and a vector $\mathbf{b} = (b_l,l\in \mc L)$. For a collision-free schedule $S$ of $\mc N$, we define a schedule $S^{\mathbf{b}}$ as
\begin{equation*}
  S^{\mathbf{b}}(l,t) = S(l,t+b_l),
\end{equation*}
which has the same rate vector as $S$. 
Then by Definition~\ref{def:cf}, we can verify that $S^{\mathbf{b}}$ is collision-free for 
$\mc N_{\mathbf{b}} = (\mc L,\mc I, D_{\mc L}^{\mathbf{b}})$, where $D_{\mc L}^{\mathbf{b}}$ is defined in \eqref{eq:db}. Due to symmetry, we can similarly argue that a collision-free schedule of $\mc N_{\mathbf{b}}$ induces a collision-free schedule of $\mc N$ of the same rate vector.
The above discussion is summarized as follows: 

\begin{proposition}\label{lemma:rate}
	For a network $\mc N$ and a vertex assignment $\mathbf{b}$,   $\mc R^{\mc N}=\mc R^{\mc N_{\mathbf{b}}}$.
\end{proposition}

Though $\mc N$ and $\mc N_{\mathbf{b}}$ are equivalent in terms of rate region, they may have different characters (see Example~\ref{ex:iad}). 
Note that the character of a network may affect the complexity for the rate region calculation according to the characterization in Sec.~\ref{sec:rr}.  
	Therefore, it is possible to use isomorphism to simplify the calculation of the rate region. In an extreme case, if $D_{\mc L}^{\mathbf{b}}$ becomes $\cv 0$, the problem is resolved. 

\begin{example}\label{ex:iad}
	Consider a network $\mc N = (\mc L,\mc I,D_{\mc L})$ with the link set $\mc L=\{l_1,l_2,l_3,l_4\}$, the collision sets
	\begin{IEEEeqnarray*}{rClrCl}
		\mc I(l_1) &=& \{l_2,l_3,l_4\}, &\quad 
		\mc I(l_2) &=& \{l_1,l_3,l_4\},\\
		\mc I(l_3) &=& \{l_2, l_4\}, &
		\mc I(l_4) &=& \{l_3\},
	\end{IEEEeqnarray*}
	and the link-wise propagation delay matrix
	\begin{equation*}
	D_{\mc L} =
	\begin{bmatrix}
	* &0&-2&-4\\
	0&*&0&-2\\
	*&0&*&0\\
	*&*&0 &*
	\end{bmatrix}.
	\end{equation*} 
	The character $D^*_{\mc N} = 4$. 	
	For the vertex assignment $\mathbf b=(4,3,2,1)$, the link-wise  delay matrix becomes
	\begin{equation*}
	D^{\mathbf b}_{\mc L} =
	\begin{bmatrix}
	* &1&0&-1\\
	-1&*&1&0\\
	*&-1&*&1\\
	*&*&-1 &*
	\end{bmatrix}.
	\end{equation*} 
	The character of $\mc N_{\mathbf b}=(\mc L,\mc I,D^{\mathbf b}_{\mc L})$ is $1$. 
\end{example}

Consider a network $\mc N = (\mc L, \mc I, D_{\mc L})$ with $D_{\mc L}\neq \cv 0$. Let $g$ be the greatest common divisor of $D_{\mc L}(l,l')$ for all $l\in \mc L$ and $l'\in \cup_{\theta\in \mc I(l)}\theta$. As we have discussed in Sec.~\ref{sec:decom}, $\mc N^\infty$ has $g$ isomorphic components. We prove that the rate region of 
${\mc N}$ is the same as the rate region of $(\mc L, \mc I, D_{\mc L}/g)$. Our proof also gives the connection of the schedules for $\mc N$ and $(\mc L, \mc I, D_{\mc L}/g)$.

\begin{proposition}
	\label{conn}
Consider a network $\mc N = (\mc L, \mc I, D_{\mc L})$ with $D_{\mc L}\neq \cv 0$. Let $g$ be the greatest common divisor of $D_{\mc L}(l,l')$ for all $l\in \mc L$ and $l'\in \cup_{\theta\in \mc I(l)}\theta$. Then,
	$\mc R^{\mc N} = \mc R^{(\mc L, \mc I, D_{\mc L}/g)}$.
\end{proposition}
\begin{IEEEproof}
  As $D_{\mc L}\neq \cv 0$, $g>0$. 
  For a collision-free schedule $S$ of $\mc N$, we define schedules $S_1, \ldots, S_g$ as:
  \begin{equation*}
    S_i(l,t) = S(l,tg+i).
  \end{equation*}
  Let's verify that $S_i$ is collision-free for $(\mc L, \mc I, D_{\mc L}/g)$.
  Suppose $S_i(l,t)=S(l,tg+i) = 1$. As $S$ is collision-free for $\mc N$, we have for any $\theta\in \mc I(l)$, there exists $l'\in \theta$ such that $S(l,tg+i+D_{\mc L}(l,l')) = S_i(l,t+D_{\mc L}(l,l')/g) =0$. Therefore, $S_i$ is collision-free for $(\mc L, \mc I, D_{\mc L}/g)$. Hence $R_S = \frac{1}{g}(R_{S_1}+\cdots+R_{S_g}) \in \mc R^{(\mc L, \mc I, D_{\mc L}/g)}$, and $\mc R^{\mc N} \subset \mc R^{(\mc L, \mc I, D_{\mc L}/g)}$.

  To prove $\mc R^{\mc N} \supset \mc R^{(\mc L, \mc I, D_{\mc L}/g)}$, consider a collision-free schedule $S'$ of $(\mc L, \mc I, D_{\mc L}/g)$. Define a schedule $S''$ for $\mc N$ such that,
  \begin{equation*}
    S''(l,ig+j) = S'(l,i), i=0,1,\ldots, j=0,1,\ldots,g-1.
  \end{equation*}
  To verify that $S''$ is collision-free for $\mc N$, consider $(l,t)$
  such that $S''(l,t)=1$. Write $t=ig+j$, where $i$ and $j$ are integers such that $i\geq 0$ and $0\leq j < g$. So $S'(l,i)=S''(l,t)=1$. Since $S'$ is collision-free for $(\mc L, \mc I, D_{\mc L}/g)$, for any $\theta\in \mc I(l)$, there exists $l'\in \theta$ such that $S''(l', ig+D_{\mc L}(l,l')+j) = S'(l',i+D_{\mc L}(l,l')/g) =0$. Hence 
  $S''(l,t)$ is collision-free for $\mc N$ as for any $\theta\in \mc I(l)$, there exists $l'\in \theta$ such that $S''(l',t+D_{\mc L}(l,l')) = 0$.
  The proof is completed as $S''$ and $S'$ have the same rate vector.
\end{IEEEproof}

In the following example, we illustrate how to combine isomorphism and connectivity to simplify a network. 

\begin{example}%
	\label{ex:9}
	Consider a network $\mc N = (\mc L, \mc I, D_{\mc L})$ with the link set $\mc L = \{l_1,l_2,l_3,l_4\}$, the collision sets
	\begin{IEEEeqnarray*}{rClrCl}
		\mc I(l_1) &=& \{l_2,l_3\}, & \quad 
		\mc I(l_2) &=& \{l_3,l_4\}, \\
		\mc I(l_3) &=& \{ l_4\}, &
		\mc I(l_4) &=& \emptyset,
	\end{IEEEeqnarray*}
	and  the link-wise propagation delay
	\begin{equation*}
	D_{\mc L} =
	\begin{bmatrix}
	* &1&5&*\\
	*&*&1&5          \\
	*&*&*&1\\
	*&*&*&*
	\end{bmatrix}.
	\end{equation*}
	This network has $D^{*}_{\mc N} = 5$. Given a vertex assignment $\mathbf{b}=[0,1,2,3]$, we get a new network  $\mc N_{\mathbf b}=(\mc L,\mc I,D^{\mathbf b}_{\mc L})$, where
	\begin{equation*}
	D^{\mathbf{b}}_{\mc L} =
	\begin{bmatrix}
	* &0&3&*\\
	*&*&0&3\\
	*&*&*&0\\
	*&*&*&*
	\end{bmatrix}.
  \end{equation*}
  The periodic graph induced by $\mc N_{\mathbf b}$ is shown in Fig.~\ref{fig:connectivity_eg}. 
	As the greatest common divisor of the relevant entries of $D^{\mathbf b}_{\mc L}$ is $3$, we have  $\mc R^{\mc N} = \mc R^{\mc N_{\mathbf b}} = \mc R^{\mc N'}$, where $\mc N'= (\mc L, \mc I, D^{\mathbf b}_{\mc L}/3)$.
	The character $D^*_{\mc N'}$ is $1$.	
\end{example}

\subsection{Scheduling with Guard Intervals}
\label{sec:framed}

Lastly in this section, we discuss the classical approach known as guarded scheduling, which involves using guard intervals to prevent collisions. While it is not necessary to be familiar with guarded scheduling in order to proceed with our approach of characterizing the scheduling rate region, this discussion can provide additional insights into the connection and distinction between  scheduling with and without delays.

Consider a network $\mc N= (\mc L, \mc I, D_{\mc L})$ with character $D^*$. 
We fix an integer $T_F\geq D^*+1$, which is called the \emph{frame length}. All timeslots $t$, $t\geq 0$, are grouped into \emph{frames}, each consisting of $T_F$ consecutive timeslots. For instance, frame $k$ ($k=0,1,\ldots$) includes timeslots $kT_F+i$, $i=0,1,\ldots,T_F-1$.
In the context of this frame structure, a schedule $S$ is considered a \emph{guarded schedule} if the last $D^*$ timeslots in each frame remain inactive. More precisely, for any frame $k$, the timeslots $kT_F + i$, $i = T_F - D^*, \ldots, T_F - 1$ are inactive. The last $D^*$ timeslots in each frame are referred to as the \emph{guard interval}.
We illustrate a guarded schedule as follows, where $T_F=5$ and $D^*=2$: 
\begin{equation*}\arraycolsep=3.6pt%
  \left[ 
  \begin{array}{c|ccccc|ccccc|ccccc|c}
    \cdots & x & x & x & 0 & 0 & x & x & x & 0 & 0 & x & x & x & 0 & 0 & \cdots
    \\
    \cdots & x & x & x & 0 & 0 & x & x & x & 0 & 0 & x & x & x & 0 & 0 & \cdots
    \\
    \cdots & x & x & x & 0 & 0 & x & x & x & 0 & 0 & x & x & x & 0 & 0 & \cdots
    \\
    \cdots & x & x & x & 0 & 0 & x & x & x & 0 & 0 & x & x & x & 0 & 0 & \cdots
  \end{array}\right]
\end{equation*}

One notable property of a guarded schedule is that it eliminates inter-frame collisions by utilizing the guard interval. This characteristic allows us to analyze the schedule in each frame independently. %
To achieve high frame efficiency, the frame length $T_F$ is typically chosen to be significantly larger than the character $D^*$. 

\subsubsection{Rate Region Approximation}

In the context of framed scheduling, the schedule of each frame,
excluding the guard interval, can be regarded as an independent set.
Consider guarded scheduling with a frame length of $T + D^*$.  Define
$\mc N^{T}$ as the subgraph of $\mc N^{\infty}$ induced by the vertex
set $\mathcal{L}\times \mc \{0,1,\ldots,T-1\}$.  An independent set of
$\mc N^{T}$ can be represented by a binary $|\mc L|\times T$ matrix.
For instance, the empty set is an independent set in $\mc N^T$
represented by the all zero $|\mc L|\times T$ matrix.  For the guarded
scheduling, each frame is collision-free if and only if the schedule
of the first $T$ time slots of the frame represents an independent set
of $\mc N^T$.

The rate vector of an independent set of $\mc N^{T}$ is
the vector obtained by summing the columns of the corresponding matrix
presentation and normalizing the result by $T$.  We define
$\dR^{\mc{N}^T}$ as the convex hull of rate vectors associated with
all independent sets of $\mathcal{N}^{T}$. Consequently, the
achievable rate region using guarded scheduling with a frame length of
$T + D^*$ is given by $\frac{T}{T + D^*}\dR^{\mc{N}^T}$, which is a
subset of $\cR^{\mc{N}}$.

\begin{proposition}\label{thm:rr1}
For a discrete network $\mc N$, $\cR^{\mc N}$ is equal to the closure of $\cup_{T=1,2,\ldots}\frac{T}{T+D^*}\dR^{\mc N^T}$.
\end{proposition}

\begin{remark}
  This characterization of the rate region $\cR^{\mc N}$ in this
  proposition involves the union of infinitely many sets, making it
  non-explicit. As the frame length $T$ increases, the approximation
  of $\mc R^{\mc N}$ by
  $\cup_{t=1,2,\ldots,T}\frac{t}{t+D^*}\dR^{\mc N^t}$ becomes more
  accurate. Nevertheless, calculating $\dR^{\mc N^T}$ using generic
  algorithms for enumerating maximal independent sets of $\mc N^T$ can
  become computationally expensive as $T$ grows. This computational complexity arises because a graph with $n$ vertices can have up to $3^{n/3}$
  maximal independent sets~\cite{moon1965cliques}.
  Therefore, although the approximation becomes more accurate with larger $T$, the computational cost of obtaining the exact characterization of $\mathcal{R}^{\mathcal{N}}$ can be prohibitive due to the exponential growth in the number of maximal independent sets for larger graphs.
\end{remark}

\begin{IEEEproof}%
  As $\cup_{T=1,2,\ldots}\frac{T}{T+D^*}\dR^{\mc N^T}\subset \cR^{\mc N}$, we only need to show $\cR^{\mc N} \subset \cup_{T=1,2,\ldots}\frac{T}{T+D^*}\dR^{\mc N^T}$.
  For any $R\in \cR^{\mc N}$ and $\epsilon>0$, by Lemma~\ref{prop:cfr}, there exists a collision-free, periodic schedule $S$ such that
  \begin{equation*}
    R_{S}(l)\geq R(l)-\epsilon \text{ for every } l\in \mc L.
  \end{equation*}
  Let $T_0\geq D^*/\epsilon$ be a period of $S$.
  We have $R_S\in \dR^{\mc N^{T_0}}$ and hence
  \begin{equation*}
    \frac{1}{1+\epsilon} (R-\epsilon) \in \frac{T_0}{T_0+D^*}\dR^{\mc N^{T_0}}.
  \end{equation*}
  As the above holds for any $\epsilon>0$, $R$ is in the closure of $\cup_{T=1,2,\ldots} \frac{T}{T+D^*}\dR^{\mc N^T}$.
\end{IEEEproof}

\subsubsection{Framed Scheduling}

Framed scheduling is a special type of guarded scheduling where each link is either active or inactive simultaneously for all the timeslots within a frame, except for the guard interval. Framed scheduling is motivated by the network scheduling schemes
extensively used in the existing wireless networks.
We illustrate a framed schedule as follows, where $T_F=5$ and $D^*=2$: 
\begin{equation*}\arraycolsep=3.6pt%
  \left[ 
  \begin{array}{c|ccccc|ccccc|ccccc|c}
    \cdots & 1 & 1 & 1 & 0 & 0 & 0 & 0 & 0 & 0 & 0 & 0 & 0 & 0 & 0 & 0 & \cdots
    \\
    \cdots & 0 & 0 & 0 & 0 & 0 & 1 & 1 & 1 & 0 & 0 & 1 & 1 & 1 & 0 & 0 & \cdots
    \\
    \cdots & 1 & 1 & 1 & 0 & 0 & 0 & 0 & 0 & 0 & 0 & 1 & 1 & 1 & 0 & 0 & \cdots
    \\
    \cdots & 0 & 0 & 0 & 0 & 0 & 1 & 1 & 1 & 0 & 0 & 0 & 0 & 0 & 0 & 0 & \cdots
  \end{array}\right]
\end{equation*}

Consider a network
$\mc N= (\mc L, \mc I, D_{\mc L})$. Recall that
$(\mc L, \mc I, \cv 0)$ is a network with the all-$0$ delay matrix and
$\cR^{(\mc L, \mc I, \cv 0)}$ can be characterized by the independent
sets of $(\mc L, \mc I)$ while ignoring the direction of edges. In other words,
$\cR^{(\mc L,\mc I,\cv 0)}$ is the convex hull of the indicator
vectors of all the independent sets of $(\mc L, \mc I)$ with the edge
directions ignored. We prove that $\cR^{(\mc L,\mc I,\cv 0)}$ is the achievable
rate region of framed scheduling for $\mc N$ when $T_F \geq 3D^*+1$.

\begin{lemma}\label{lemma:cf1}
	A framed schedule $S$ with a frame length $T_F\geq 3D^*+1$ is collision-free if and only if for any link $l$ that is active in a frame, for all $\theta\in \mc I(l)$, there exists a certain $l'\in \theta$ that is inactive in the same frame. 
\end{lemma}
\begin{IEEEproof}%
  A framed schedule is also a guarded schedule, where a collision can only be generated by links within the same frame. Therefore, the sufficiency of the lemma holds (even without the condition that $T_F\geq 3D^*+1$). 
  
  To prove the necessary condition, consider that link $l$ is active in the first frame, and for a certain $\theta\in \mc I(l)$, all $l'\in \theta$ are active in the first frame. As $|D_{\mc L}(l,l')|\leq D^*$ for all $l'\in \theta$, we have $D^*+D_{\mc L}(l,l')\in \{ 0,1,\ldots, 2D^*\}$ for all $l'\in \theta$. As $T_F \geq 3D^*+1$, for all $l'\in \theta$, $S(l',D^*+D_{\mc L}(l,l'))=1$, and hence $S(l,D^*)$ has a collision.
\end{IEEEproof}

Based on the above lemma, the following statement is straightforward. 

\begin{proposition}\label{thm:framed}
	Consider a network $(\mc L,\mc I, D_{\mc L})$.
	For a framed schedule of frame length $T_F\geq 3D^*+1$, if its rate vector exists, the rate vector is in $(1-D^*/T_F) \cR^{(\mc L,\mc I,\cv 0)}$. Moreover, any rate vector in $\cR^{(\mc L,\mc I,\cv 0)}$ can be achieved by collision-free, framed schedules.
\end{proposition}
\begin{IEEEproof}[Proof outline]
First, the scheduling rate within each frame is in $(1-D^*/T_F) \cR^{(\mc L,\mc I,\cv 0)}$. Therefore, if the average rate of all the frames converge, it must be also in $(1-D^*/T_F) \cR^{(\mc L,\mc I,\cv 0)}$. Second, an independent set of $(\mc L,\mc I)$ with edge direction ignored can be used to design a collision-free framed schedule, and hence the achievable part can be shown using a large $T_F$. 
\end{IEEEproof}

For a network with a binary collision profile (i.e., $\mc I(l)$ is a subset of $\mc L$ for all $l\in \mc L$), the above discussion for framed scheduling can be improved by relaxing the condition $T_F \geq 3D^*+1$ to $T_F\geq 2D^*+1$ in Lemma~\ref{lemma:cf1} and Proposition~\ref{thm:framed}. The proof of the  necessary condition of Lemma~\ref{lemma:cf1} can be modified as follows: Consider a certain $l'\in \mc I(l)$ is active in the first frame. As $|D_{\mc L}(l,l')|\leq D^*$, there exists $t_0\in \{ 0,1,\ldots, D^*\}$ such that $t_0+D_{\mc L}(l,l')\in \{ 0,1,\ldots, D^*\}$. As $T_F \geq 2D^*+1$, 
 $S(l',t_0+D_{\mc L}(l,l'))=1$, and hence  $S(l,t_0)$ has a collision.

\section{Scheduling Graphs and Rate Region}
\label{sec:sg}

The characterization of the scheduling rate region $\mc R^{\mc N}$
using guarded scheduling (Proposition~\ref{thm:rr1}) cannot be exactly
computed in finite time. In this section, we provide an explicit
characterization of $\mc R^{\mc N}$ (Theorem~\ref{thm:reg-1} and
Theorem \ref{thm:reg-3} below) that enables the computation of this
region in finite time.  Our approach leverages the periodic structure
of $\mc N^{\infty}$.

We need some further concepts about directed graphs: In a directed graph $\mc G$, a \emph{path} of length $k$ is a sequence of vertices $v_0,v_1,\ldots, v_k$ where $(v_i,v_{i+1})$ ($i=0,1,\ldots, k-1$) is a directed edge in $\mc G$. A path of length $0$ is a vertex, while a path of length $1$ is an edge. A path $(v_0,v_1,\ldots, v_k)$ is said to be \emph{closed} if $v_k=v_0$. A path $(v_0,v_1,\ldots)$ of infinite length is said to be \emph{periodic} if there exists a positive integer $T$ such that $v_i = v_{i+T}$ for any $i\geq 0$, where each such value of $T$ is called a period. 
For a periodic path with a period of $T$, the sub-path $(v_0,v_1,\ldots, v_{T})$ is closed. A \emph{cycle} in $\mc G$ is a closed path $(v_0,v_1,\ldots,v_k)$ where $v_i \neq v_j$ for any $0\leq i \neq j \leq k-1$. In other words, the only repeated vertices in the cycle are the first and the last vertices. A cycle of length $k$ is also called a $k$-cycle.

\subsection{Scheduling Graphs and Collision-free Schedules}
\label{sec:sgcs}

Recall that a schedule $S$ is a matrix with columns indexed by
$t\in\mathbb{Z}$. We begin by dividing $S$ into submatrices, which are
formed by consecutive columns, and then proceed to verify whether $S$
is collision-free using these submatrices. For integers $T$, $Q$, and
$k$ satisfying $T\geq 1$ and $1\leq Q\leq T$, we denote $S[T,Q,k]$ as
the submatrix of $S$ with columns $kQ, kQ+1, \ldots, kQ+T-1$.  We
refer to $T$ as the \emph{blocklength}, $Q$ as the \emph{step size},
and $k$ as the \emph{block index}. The definition is illustrated in
Fig.~\ref{fig:la}. Note that $S[T,Q,0]=S[T,T,0]$
for all $Q$. When $Q<T$, there is overlap between $S[T,Q,k]$ and
$S[T,Q,k+1]$. However, $S[T,T,k]$ and $S[T,T,k+1]$ are disjoint but
adjacent.

\begin{figure}
  \centering
  \subfigure[{$S[3,2,k]$ for $k=0,1,2$}]{
  \begin{tikzpicture}[font=\footnotesize,link/.style={circle, draw, thin,fill=cyan!20},scale=0.55,
    abox/.style={draw=red!50!black,dotted,thick,inner sep=4},
    bbox/.style={draw=green!80!black,dashed,thick,inner sep=4}
    ]

      \foreach \x in {1,2,3,4}
      {
        \foreach \y in {0,1,2,3,4,5,6}
        {
          \node[link] (a\x\y) at (2.3*\y,-1.8*\x) {};
        }
        \node[left=2mm of a\x0] {$l_{\x}$};
        \node[right=3mm of a\x6] {$\ldots$};
      }

      \foreach \y in {0,1,2,3,4,5,6}
      {
        \draw[->, bend left=45] (a1\y) to (a3\y);
        \draw[->, bend left=45] (a2\y) to (a4\y);
      }

      \foreach \y/\z in {0/1,1/2,2/3,3/4,4/5,5/6}
      {
        \draw[->] (a1\y) -- (a2\z);
        \draw[->] (a2\y) -- (a3\z);
        \draw[->] (a3\y) -- (a4\z);
      }

      \node[abox,fit=(a10) (a42),label=below:{$S[3,2,0]$}] {};
      \node[bbox,fit=(a12) (a44),label=below:{$S[3,2,1]$}] {};
      \node[abox,fit=(a14) (a46),label=below:{$S[3,2,2]$}] {};
    \end{tikzpicture}}

  \subfigure[{$S[2,2,k]$ for $k=0,1,2$}]{
    \begin{tikzpicture}[font=\footnotesize,link/.style={circle, draw, thin,fill=cyan!20},scale=0.55,
    abox/.style={draw=red!50!black,dotted,thick,inner sep=4},
    ]

      \foreach \x in {1,2,3,4}
      {
        \foreach \y in {0,1,2,3,4,5,6}
        {
          \node[link] (a\x\y) at (2.3*\y,-1.8*\x) {};
        }
        \node[left=2mm of a\x0] {$l_{\x}$};
        \node[right=3mm of a\x6] {$\ldots$};
      }

      \foreach \y in {0,1,2,3,4,5,6}
      {
        \draw[->, bend left=45] (a1\y) to (a3\y);
        \draw[->, bend left=45] (a2\y) to (a4\y);
      }

      \foreach \y/\z in {0/1,1/2,2/3,3/4,4/5,5/6}
      {
        \draw[->] (a1\y) -- (a2\z);
        \draw[->] (a2\y) -- (a3\z);
        \draw[->] (a3\y) -- (a4\z);
      }

      \node[abox,fit=(a10) (a41),label=below:{$S[2,2,0]$}] {};
      \node[abox,fit=(a12) (a43),label=below:{$S[2,2,1]$}] {};
      \node[abox,fit=(a14) (a45),label=below:{$S[2,2,2]$}] {};
    \end{tikzpicture}
    }
    \caption{Illustration of the associated part in the periodic graph of  $S[T,Q,k]$.}\label{fig:la}
\end{figure}

For a positive integer $T$ and a network $\mc N$, an
$|\mc L|\times T$ binary matrices $A$ is considered collision-free for $\mc N$ if
$A = S'[T,T,0]$ for some collision-free schedule $S'$, or equivalently, $A$ represents an independent set of $\mc N^{T}$.  Fix an integer $Q$ with
$1\leq Q \leq T$.  If a schedule $S$ is collision-free, then
$S[T,Q,k]$, $k=0,1,\ldots$ are all collision-free. Conversely,
we will show that for sufficiently large $T$, a schedule $S$ is
collision-free if $(S[T,Q,k],S[T,Q,k+1])$, $k=0,1,\ldots$ satisfy a certain condition. To present this condition, we adopt a graphical approach that enables us to leverage results from graph theory conveniently.
A pair of matrices $(A_1,A_2)$, where
$A_i\in \{0,1\}^{|\mc L|\times t_i}$, is also regarded as a matrix obtained
by juxtaposing them. 

\begin{definition}[Scheduling graph]
\label{def:schedule_graph}
  For a network $\mc N$ and integers $1\leq Q\leq T$, a \emph{scheduling graph} is a directed graph, denoted by $(\mc M_{T}, \mc E_{T,Q})$, defined as follows: The vertex set $\mc M_{T}$ consists of all $|\mc L|\times T$ binary matrices that are collision-free for $\mc N$. The edge set $\mc E_{T,Q}$ includes all pairs of  vertices $(A,B)$ such that $A[T-Q,Q,1] = B[T-Q,Q,0]$ and $(A[Q,Q,0],B)$ (considered as an $|\mc L|\times T$ matrix) is collision-free for $\mc N$.
\end{definition}

The sets $\mc M_T$ and $\mc E_{T,Q}$ can be determined by the independent sets of $\mc N^T$ and  $\mc N^{T+Q}$, respectively, as discussed in more detail in Sec.~\ref{sec:sgcal}. We also call $(\mc M_T, \mc E_{T,Q})$ the  \emph{step-$Q$} scheduling graph. 
When $Q<T$, a necessary condition for $(A,B) \in \mc E_{T,Q}$ is that the last $T-Q$ columns of $A$ and the first $T-Q$ columns of $B$ are the same.
Moreover, $(A,B) \in \mc E_{T,Q}$ if and only if $A = S[T,Q,0]$ and $B=S[T,Q,1]$ for a certain collision-free schedule $S$.
We also write the step-$T$ scheduling graph as $(\mc M_T, \mc E_{T})$.
A necessary and sufficient condition for $(A,B) \in \mc E_{T}$ is that 
 $(A,B)$ as an $|\mc L|\times 2T$ binary matrix represents an independent set of $\mc N^{2T}$.  %

\begin{example}[Multihop line network]
  \label{ex:scheduling_graph}
  We give $(\mc M_1, \mc E_{1})$ of  $\mc N_{4,1}^{\text{line}}$ as an example. Here $\mc M_1$ includes the $4\times 1$ matrices $v$ such that $v$ can be a column of a certain collision-free schedule $S$ of $\mc N_{4,1}^{\text{line}}$. We have $\mc M_1 = \{v_0,v_1,\ldots,v_8\}$, where
  \begin{equation*}
    \begin{bmatrix}
      v_0 & v_1 & \cdots & v_8
    \end{bmatrix}
    =
    \begin{bmatrix}
      0 & 1 & 0 & 0 & 0 & 1 & 1 & 0 & 0 \\
      0 & 0 & 1 & 0 & 0 & 0 & 1 & 1 & 0 \\
      0 & 0 & 0 & 1 & 0 & 0 & 0 & 1 & 1 \\
      0 & 0 & 0 & 0 & 1 & 1 & 0 & 0 & 1
    \end{bmatrix}.
  \end{equation*}
  $\mc E_{1}$ includes all the pairs $(v,v')$ such that $[v,v']$ is equal to two consecutive columns of a certain collision-free schedule, and can be denoted by the adjacency matrix:
  \begin{equation}\label{eq:ee}
    \begin{blockarray}{cccccccccc}
      &v_0&v_1&v_2&v_3&v_4&v_5&v_6&v_7&v_8\\
      \begin{block}{c[ccccccccc]}
	v_0&1&1&1&1&1&1&1&1&1\\
	v_1&1&1&0&1&1&1&0&0&1\\
	v_2&1&1&1&0&1&1&1&0&0\\
	v_3&1&1&1&1&0&0&1&1&0\\
	v_4&1&1&1&1&1&1&1&1&1\\
	v_5&1&1&0&1&1&1&0&0&1\\
	v_6&1&1&0&0&1&1&0&0&0\\
	v_7&1&1&1&0&0&0&1&0&0\\
	v_8&1&1&1&1&0&0&1&1&0\\
      \end{block}
    \end{blockarray}.
  \end{equation}
\end{example}

The following two theorems show that a collision-free schedule of a
network~$\mc N$ is equivalent to a directed path in a scheduling graph
$(\mc M_T, \mc E_{T,Q})$ with $T\geq 2D^*$. 
These results allow us to further investigate the scheduling problem using  scheduling graphs.

\begin{theorem}
\label{thm:cf-1}
Consider a network $\mc N$ and a schedule $S$.  If $S$ is
collision-free for $\mc N$, then for any integers $1\leq Q\leq T$, the
sequence $(S[T,Q,k], k=0,1,\ldots)$ forms a path in
$(\mc M_T,\mc E_{T,Q})$.
\end{theorem}

\begin{IEEEproof}%
  Suppose $\mathcal{S}$ is collision-free. We see that for $k=0,1,\ldots$, $S[T,Q,k]$ is collision-free for $\mc N$ and hence is in $\mc M_T$. Note that $S[T,Q,k]$ and  $S[T,Q,k+1]$ are the first and the last $T$ columns of $S[T+Q,Q,k]$, which is collision-free for $\mc N$. Hence, $(S[T,Q,k],S[T,Q,k+1]) \in \mc E_{T,Q}$. 
  Therefore,
$(S[T,Q,k], k=0,1,\ldots)$ is a path in  $(\mc M_T,\mc E_{T,Q})$.
\end{IEEEproof}  

For $\mc N_{4,1}^{\text{line}}$, any schedule $S$ that forms a path in
$(\mc M_1, \mc E_{1})$ as characterized in
Example~\ref{ex:scheduling_graph} is collision-free. However, 
the converse of Theorem~\ref{thm:cf-1} can only be proved in general when $T$ is sufficiently large. 
The next example shows that for $T<2D^*$, a schedule $S$ that forms a path in $(\mc M_T, \mc E_{T,Q})$ may not be collision-free.

\begin{example}
Consider a network $\mc N_4 = (\mc L, \mc I, D_{\mc L})$, where
$\mc L = \{l_1,l_2,l_3,l_4\}$.
The collision sets of the links are
  \begin{IEEEeqnarray*}{rClrCl}
	\mc I(l_1) & = & \emptyset, &
	\mc I(l_2) & = & \{ \{l_1,l_3\}\}, \\
	\mc I(l_3) & = & \{ \{l_2,l_4\}\}, & \quad
	\mc I(l_4) & = & \emptyset.
\end{IEEEeqnarray*}
The link-wise delay matrix $D_{\mc L}$ is
\begin{equation}
D_{\mc L} =
\begin{bmatrix}
* &*&*&*\\
-1&*&1&*\\
*&-1&*&1\\
*&*&*&*
\end{bmatrix}.
\end{equation}
For this network, the character
 \begin{equation}\label{eq:dstarex}
D^* = \max_{l\in \mc L}\max_{\theta\in \mc I(l)} \max_{l'\in \theta}|D_{\mc L}(l,l')| = 1.
\end{equation}
We illustrate that a schedule $S$ that forms a path in $(\mc M_1, \mc E_{1})$ may not be collision-free. First, we see that $(\mc M_1, \mc E_{1})$ is a complete graph with the vertices set $\{0,1\}^4$.
Consider a schedule $S$ with a submatrix $S'$ formed by three consecutive columns:
\begin{align*}
  S' =\begin{bmatrix}
    1   &0 &  0\\  1& 1& 0\\  0 &1& 1\\  0& 0& 1
  \end{bmatrix}.
\end{align*} 
Because $S'(l_1,0) = S'(l_3,2) = 1$, $S'(l_2,1)$ has a collision. As $S'(l_2,1)=1$, $S$ is not collision-free.
\end{example}

The next theorem proves a converse of Theorem~\ref{thm:cf-1} for blocklength $T \geq 2D^*$. For binary collision, the converse of can be proved for $T\geq D^*$ (see Theorem~\ref{thm:cf-2}).

\begin{theorem}\label{thm:cf-3}
  Consider a network $\mc N=(\mc L, \mc I, D_{\mc L})$ and a schedule
  $S$.  If for certain integers $T$ and $Q$ such that $T\geq 2D^*$ and
  $T\geq Q\geq 1$, the sequence $(S[T,Q,k], k=0,1,\ldots)$ forms a
  path in $(\mc M_T,\mc E_{T,Q})$, then $S$ is collision-free.
\end{theorem}
\begin{IEEEproof}%
  Fix any $(l,t)\in \mc L\times \zz$ such that $S(l,t)=1$, and fix any $\theta \in \mc I(l)$. To prove $S(l,t)$ is collision-free, we need to show that for a certain $l'\in \theta$ $S(l',t+D_{\mc L}(l,l')) = 0$. 
  Find integers $c\geq 0$ and $0\leq d < Q$ such that $D^* = c Q + d$, and
  find integers $k$ and $0\leq t_0 < Q$ such that $t=(k+c)Q+t_0$. For any $l' \in \theta$, we have
	\begin{IEEEeqnarray*}{rCl}
		t' & \triangleq & t+D_{\mc L}(l,l')  \\
		& \in & [t - D^*, t+D^*] \\
        & = & [kQ + t_0 -d, kQ+t_0-d+2D^*]. %
	\end{IEEEeqnarray*}
	In the following, we discuss two cases of $t_0$: $0\leq t_0<d$ and $d \leq t_0 <Q$. %

    When $0\leq t_0< d$, we have $-Q<t_0-d<0$. Hence for any $l' \in \theta$,
    as $T\geq 2D^*$, $(k-1)Q < t' < kQ+T$ (see Fig.~\ref{fig:cf-3}
    (a)). In other words, to verify the collision of $S(l,t)$ with
    respect to $\theta$, we only need to consider $S[T,Q, k-1]$ and
    $S[T,Q,k]$.  As $(S[T,Q, k-1],S[T,Q,k]) \in \mc E_{T,Q}$, we have
    $S[T,Q,k-1] = S'[T,Q,0]$ and $S[T,Q,k] = S'[T,Q,1]$ for a certain
    collision-free schedule $S'$. As $S(l,t) = S'(l,t_0-d+Q+D^*)=1$,
    we have $S(l',t')=S'(l',t_0-d+Q+D^*+D_{\mc L}(l,l')) = 0$ for a
    certain $l'\in \theta$.
    
    When $d \leq t_0 < Q$, we have $0\leq t_0-d<Q$. Hence for any
    $l' \in \theta$, as $T\geq 2D^*$, $kQ\leq t' < (k+1)Q+ T$ (see
    Fig.~\ref{fig:cf-3} (b)). In other words, to verify the
    collision of $S(l,t)$ with respect to $\theta$, we only need to
    consider $S[T,Q, k]$ and $S[T,Q,k+1]$. As
    $(S[T,Q,k],S[T,Q,k+1]) \in \mc E_{T,Q}$, we have
    $S[T,Q,k] = S'[T,Q,0]$ and $S[T,Q,k+1] = S'[T,Q,1]$ for certain
    collision-free schedule $S'$. Therefore, as
    $S(l,t) = S'(l,t_0-d+D^*)=1$, we have
    $S(l',t')=S'(l',t_0-d+D^*+D_{\mc L}(l,l')) = 0$ for a certain
    $l'\in \theta$.

	For both cases, $S(l',t') = 0$ for a certain $l'\in \theta$.
	Therefore, $S(l,t)$ is collision-free. 
\end{IEEEproof}

\begin{figure}
  \centering
  \subfigure[case $0\leq t_0<d$]{
  \begin{tikzpicture}[font=\footnotesize]
    \draw[->] (0,0) -- (8,0) node[right] {time};
    \draw[thick] (0.5,0) -- +(0,0.2) node[above] {$(k-1)Q$};
    \draw[thick] (1.5,0) -- +(0,0.2) node[above] {$kQ$};
    \draw[thick] (4,0) -- +(0,0.2) node[above] {$(k+c)Q$};
    \draw[] (7.6,0) -- +(0,0.2) node[above] {$kQ+T$};
    \draw[] (4.3,0) -- +(0,-0.2) node[below] {$t$};
    \draw[] (1.2,0) -- +(0,-0.2) node[below] {$t-D^*$};
    \draw[] (7.4,0) -- +(0,-0.2) node[below] {$t+D^*$};
   \end{tikzpicture}
   }

  \subfigure[case $d\leq t_0<Q$]{
  \begin{tikzpicture}[font=\footnotesize,
    adot/.style={circle,draw,thick,fill=white,inner sep=0,minimum size=3pt},
    bdot/.style={fill,inner sep=0,minimum size=3pt}]
    \draw[->] (0,0) -- (8,0) node[right] {time};
    \draw[thick] (0.5,0) -- +(0,0.2) node[above] {$(k-1)Q$};
    \draw[thick] (1.5,0) -- +(0,0.2) node[above] {$kQ$};
    \draw[thick] (4,0) -- +(0,0.2) node[above] {$(k+c)Q$};
    \draw[] (7.6,0) -- +(0,0.2) node[above] {$(k+1)Q+T$};
    \draw[] (4.3,0) -- +(0,-0.2) node[below] {$t$};
    \draw[] (1.7,0) -- +(0,-0.2) node[below] {$t-D^*$};
    \draw[] (6.9,0) -- +(0,-0.2) node[below] {$t+D^*$};
  \end{tikzpicture}
   }
  \caption{Illustration of the proof of Theorem~\ref{thm:cf-3}. A thick tick indicts the start position of a submatrix $S[T,Q,k]$, and a thin tick indicts the time.} \label{fig:cf-3}
\end{figure}

\subsection{Periodic Schedules and Scheduling Graphs}
\label{sec:pssg}

Theorem~\ref{thm:cf-1} and Theorem~\ref{thm:cf-3} together show that a
collision-free schedule is equivalent to a directed path in a
scheduling graph $(\mc M_T, \mc E_{T,Q})$ with $T\geq 2D^*$. Hence
we convert the independent set problem on a periodic hypergraph to a path problem on a scheduling graph. Although a scheduling graph has a finite size, 
the number of paths in it is infinite and the length of a path can be unbounded as well. 
We continue to study how to reduce the number and length of the paths based on periodic scheduling, which is rate region achieving as shown in Lemma~\ref{prop:cfr}.
We will establish a relationship between a periodic schedule and a closed path in a scheduling graph, thus providing a characterization of $\cR^{\mc N}$ by cycles in the scheduling graph.

\begin{definition}\label{def:rcp}
  The rate vector $R_P$ of a closed path $P=(A_0,A_1,\ldots,A_k)$ in $(\mc M_T,\mc E_{T,Q})$ is defined as follows:
\begin{equation*}
  R_P = \frac{1}{kQ} \sum_{i=0}^{k-1} A_i\cv 1_Q,
\end{equation*}
where $\cv 1_Q$ is a length-$T$ column vector with the first $Q$ entries equal to $1$ and the remaining $T-Q$ entries equal to $0$. 
\end{definition}
Denote $\mathrm{cycle}(\mc G)$ as the collection of all cycles in a directed graph $\mc G$.
Define
\begin{equation*}
  \mc R^{(\mc M_T,\mc E_{T,Q})} = \conv\{R_C: C \in \mathrm{cycle}(\mc M_T, \mc E_{T,Q})\},
\end{equation*}
where $\conv\mc A$ is the convex hull of a set $\mc A$. 
Since $(\mc M_T,\mc E_{T,Q})$ is finite, $\mathrm{cycle}(\mc M_T,\mc E_{T,Q})$ is finite and hence $\mc R^{(\mc M_T,\mc E_{T,Q})}$ is a closed convex polytope. We will show that when $T\geq 2D^*$, $\cR^{\mc N} = \mc R^{(\mc M_T,\mc E_{T,Q})}$.

\begin{lemma}\label{thm:m1e1}
  For a network $\mc N$ and a collision-free, periodic schedule $S$, the following statements hold:
  \begin{enumerate}
  \item For any period $K$ of $S$, $(S[T,Q,i], i=0,1,\ldots,K)$ forms a closed path in the scheduling graph $(\mc M_T, \mc E_{T,Q})$. 
  \item The rate vector $R_S$ belongs to the set $\mc R^{(\mc M_T,\mc E_{T,Q})}$.
  \end{enumerate}
\end{lemma}
\begin{IEEEproof}%
	By Theorem \ref{thm:cf-1}, $(S[T,Q,i], i=0,1,\ldots)$ forms a path in $(\mc M_T,\mc E_{T,Q})$. As $KQ$ is also a period of $S$, $S[T,Q,i] = S[T,Q,i+K]$. Therefore, the path $(S[T,Q,i], i=1,2,\ldots)$ has a period $K$ and hence $(S[T,Q,i], i=0,1,\ldots,K)$ is a closed path in $(\mc M_T, \mc E_{T,Q})$.

	A closed path can be decomposed into a sequence of (not necessarily distinct) cycles (see, e.g., \cite{gleiss2003circuit}). Suppose $(S[T,Q,i], i=0,1,2,\ldots,K)$ has the decomposition of cycles $C_1,\ldots,C_{K'}$ in $\mathrm{cycle}(\mc M_T, \mc E_{T,Q})$, where $C_i$ is of length $k_i$. Using this decomposition of the closed path, one obtains
	\begin{IEEEeqnarray*}{rCl}
		R_S %
		& = & \frac{1}{K} \sum_{i=1}^{K'}k_i R_{C_i} 
		 \in  \mc R^{(\mc M_T,\mc E_{T,Q})}. \qedhere
	\end{IEEEeqnarray*}
\end{IEEEproof}

\begin{theorem}\label{thm:reg-1}
For a network $\mc N$, we have $\cR^{\mc N} \subset \mc R^{(\mc M_T,\mc E_{T,Q})}$ for integers $T\geq Q \geq 1$.
\end{theorem}

\begin{IEEEproof}%
Consider $R \in \cR^{\mc N}$. By Lemma~\ref{prop:cfr}, for any $\epsilon >0$, there exists a collision-free, periodic schedule $S$ such that $R_S\pgeq R-\epsilon$. By Lemma~\ref{thm:m1e1}, $R_S \in \mc R^{(\mc M_T,\mc E_{T,Q})}$. As $\mc R^{(\mc M_T,\mc E_{T,Q})}$ is closed, we have $R \in  \mc R^{(\mc M_T,\mc E_{T,Q})}$.
\end{IEEEproof}

For the general collision model, the converse of the above theorem ($\mc R^{(\mc M_T,\mc E_{T,Q})} \subset \cR^{\mc N}$) can be proved for blocklength $T\geq 2D^*$. To show the converse, we first show that for a periodic schedule, it is sufficient to check a sufficiently long part of the schedule to verify whether it is collision-free.

\begin{lemma}\label{lemma:cfp-2}
	For a network $\mc N$ and integers $T\geq 2D^*$ and $1\leq Q\leq T$, 
	suppose $S$ is a periodic schedule with period $kQ$ such that $(S[T,Q,i], i=0,\ldots,k)$ is a closed path in $(\mc M_T, \mc E_{T,Q})$. Then $S$ is collision-free for $\mc N$.
\end{lemma}
\begin{IEEEproof}%
  Let $A_i = S[T,Q,i]$ for $i=0,1,\ldots,k-1$.   
Fix any integer $i = a k + b$ where $a \geq  0$ and $0\leq b \leq k-1$. First, $S[T,Q,i] = S[T,Q,b] = A_b \in \mc M_T$. Second, $(S[T,Q,i], S[T,Q,i+1]) = (S[T,Q,b], S[T,Q,b+1]) = (A_b,A_{b+1})\in \mc E_{T,Q}$. Therefore, the sequence $(S[T,Q,i], i=1,2,\ldots)$ is a path in  $(\mc M_T,\mc E_{T,Q})$. As $T\geq 2D^*$, by Theorem~\ref{thm:cf-3}, $S$ is collision-free.
\end{IEEEproof}

\begin{theorem}\label{thm:reg-3}
  For a network $\mc N$ with a general collision profile and any integers $T$ and $Q$ such that $T\geq 2D^*$ and 
and $1\leq Q\leq T$, it holds that 
    $\cR^{\mc N} \supset \mc R^{(\mc M_T,\mc E_{T,Q})}$.
\end{theorem}
\begin{IEEEproof}%
	Fix $R \in  \mc R^{(\mc M_T,\mc E_{T,Q})}$. We can write
	\begin{equation*}
	R = \sum_{C\in \mathrm{cycle}(\mc M_T,\mc E_{T,Q})} \alpha_C R_C,
	\end{equation*}
	where $\alpha_C\geq 0$ and $\sum_{C\in \mathrm{cycle}(\mc M_T,\mc E_T)}\alpha_C = 1$. For a cycle $C=(C_0,C_1,\ldots, C_{k})$ in $(\mc M_T,\mc E_{T,Q})$, we define a schedule $S$ with period $kQ$ such that $S[T,Q,i] = C_i$ for $i=0,1,\ldots,k-1$. By Lemma~\ref{lemma:cfp-2}, $S$ is collision-free and hence $R_C = R_S \in \cR^{\mc N}$. As $\cR^{\mc N}$ is convex (see Lemma~\ref{lemma:convex}), we have $R \in \cR^{\mc N}$.  
\end{IEEEproof}

Theorem~\ref{thm:reg-1} and Theorem \ref{thm:reg-3} together give an explicit characterization of $\cR^{\mc N}$, i.e., when $T\geq 2D^*$, $\cR^{\mc N} = \mc R^{(\mc M_T,\mc E_{T,Q})}$, 
where $\mc R^{(\mc M_T,\mc E_{T,Q})}$ is explicitly determined by the cycles in $(\mc M_T,\mc E_{T,Q})$. Now we see that using a scheduling graph, only a finite number of cycles are required for determining the scheduling rate region.
Moreover, from the proofs of Lemma~\ref{thm:m1e1} and Lemma~\ref{lemma:cfp-2}, we also see that a periodic, collision-free schedule can be formed by  cycles of $(\mc M_T,\mc E_{T,Q})$.

\subsection{Enhanced Results for Binary Network Model}
\label{sec:bin}

For the binary collision model, Theorem~\ref{thm:reg-3} can be proved
for $T\geq D^*$ (see Theorem \ref{thm:reg-2} below).  For the binary
collision model, the collision set $\mc I(l)$ has the property that
for any $\theta \in \mc I(l)$, $|\theta|=1$. In this case, we also
write $\mc I(l)$ as a subset of $\mc L$, and the formula of character $D^*$ given in
\eqref{eq:dstar} can be
simplified as
\begin{equation*}
D^*_{\mc N} = \max_{l\in \mc L}\max_{l'\in \mc I(l)} |D_{\mc L}(l,l')|.
\end{equation*}

The following theorem improves Theorem~\ref{thm:cf-3} for the binary collision model with the lower bound on $T$ improved from $2D^*$ to $D^*$.

\begin{theorem}
	\label{thm:cf-2}
	Consider a network $\mc N=(\mc L, \mc I, D_{\mc L})$ with a binary
    collision profile and a schedule $S$.  If for certain integers $T$
    and $Q$ such that $T\geq D^*$ and $T\geq Q\geq 1$, the sequence
    $(S[T,Q,k], k=0,1,\ldots)$ forms a path in
    $(\mc M_T,\mc E_{T,Q})$, then $S$ is collision-free.
\end{theorem}
\begin{IEEEproof}
  Fix any $(l,t)\in \mc L\times \zz$ such that $S(l,t)=1$, and fix any $l' \in \mc I(l)$. To prove $S(l,t)$ is collision-free, we need to show that $S(l',t+D_{\mc L}(l,l')) = 0$. 
  Find integers $c\geq 0$ and $0\leq d < Q$ such that $D^* = c Q + d$, and find integers $k$ and $0\leq t_0 < Q$ such that $t=(k+c)Q+t_0$. For $l'$, we have
  \begin{IEEEeqnarray*}{rCl}
    t' & \triangleq & t+D_{\mc L}(l,l')  \\
    & \in & [t - D^*, t+D^*] \\
    & = & [kQ + t_0 -d, kQ+t_0-d+2D^*]. %
  \end{IEEEeqnarray*}
  In the following, we discuss two cases of $t_0$: $0\leq t_0<d$ and $d \leq t_0 <Q$. 

    When $0\leq t_0< d$, we have $-Q<t_0-d<0$, and hence $(k-1)Q < t' < kQ+2D^*$.
    See Fig.~\ref{fig:cf-2} (a) for an illustration. 
	Consider three subcases of $t'$:
	\begin{itemize}
    \item  $t'\in [(k-1)Q, kQ+T-1]$. Note that $t \in [kQ,kQ+T-1]$.
      As $(S[T,Q,k-1],S[T,Q,k]) \in \mc E_{T,Q}$, we have
      $S[T,Q,k-1] = S'[T,Q,0]$ and $S[T,Q,k] = S'[T,Q,1]$ for certain collision-free schedule $S'$. Therefore, as $S(l,t) = S'(l,t_0+(1+c)Q)=1$, $S(l',t')=S'(l',t_0+(1+c)Q+D_{\mc L}(l,l')) = 0$.
      
    \item  $t'\in [(k+c)Q, (k+c)Q+T-1]$.  Note that $t \in [(k+c)Q, (k+c)Q+T-1]$ too. As $S[T,Q,k+c] \in \mc M_T$, we have $S[T,Q,k+c] = S'[T,0]$ for certain collision-free schedule $S'$. Therefore, as $S(l,t) = S'(l,t_0)=1$, $S(l',t')=S'(l',t_0+D_{\mc L}(l,l')) = 0$. 

    \item $t'\in [(k+c+1)Q, (k+c+1)Q+T-1]$. Note that $t \in [(k+c)Q, (k+c)Q+T-1]$.
      As $(S[T,Q,k+c],S[T,Q,k+c+1])\in \mc E_{T,Q}$, we have $S[T,Q,k+c] = S'[T,Q,0]$ and $S[T,Q,k+c+1] = S'[T,Q,1]$ for certain collision-free schedule $S'$. Therefore, as $S(l,t) = S'(l,t_0)=1$, $S(l',t')=S'(l',t_0+D_{\mc L}(l,l')) = 0$.
	\end{itemize}
We see that for all the subcases of $t'$, $S(l',t')=0$.
    
    When $d \leq t_0 < Q$, we have $0\leq t_0-d<Q$, and hence 
    $kQ\leq t' < kQ+ T+Q$. See Fig.~\ref{fig:cf-2} (b) for an illustration. 
    Consider three subcases of $t'$:
    \begin{itemize}
    \item $t'\in [kQ, kQ+T-1]$. Note that $t \in [kQ,kQ+T-1]$.
    \item $t'\in [(k+c)Q, (k+c)Q+T-1]$.   Note that $t \in [(k+c)Q, (k+c)Q+T-1]$ too. 
    \item $t'\in [(k+c+1)Q, (k+c+1)Q+T-1]$. Note that $t \in [(k+c)Q, (k+c)Q+T-1]$.
    \end{itemize}
    These subcases can be analyzed similarly as when $0\leq t_0<d$, and hence $S(l',t')=0$.

	For both cases of $t_0$, $S(l',t') = 0$ for any $l' \in \mc I(l)$.
	Therefore, $S(l,t)$ is collision-free. 
      \end{IEEEproof}

      \begin{figure}
  \centering
  \subfigure[case $0\leq t_0<d$]{
  \begin{tikzpicture}[font=\footnotesize]
    \draw[->] (0,0) -- (8,0) node[right] {time};
    \draw[thick] (0.5,0) -- +(0,0.2) node[above] {$(k-1)Q$};
    \draw[thick] (1.5,0) -- +(0,0.2) node[above] {$kQ$};
    \draw[thick] (3.3,0) -- +(0,0.2) node[above] {$(k+c)Q$};
    \draw[] (4.6,0) -- +(0,0.2) node[above] {$kQ+T$};
    \draw[] (6.6,0) -- +(0,0.2) node[above] {$(k+c+1)Q+T$};
    \draw[] (3.5,0) -- +(0,-0.2) node[below] {$t$};
    \draw[] (1.2,0) -- +(0,-0.2) node[below] {$t-D^*$};
    \draw[] (5.8,0) -- +(0,-0.2) node[below] {$t+D^*$};
   \end{tikzpicture}
   }

  \subfigure[case $d\leq t_0<Q$]{
  \begin{tikzpicture}[font=\footnotesize,
    adot/.style={circle,draw,thick,fill=white,inner sep=0,minimum size=3pt},
    bdot/.style={fill,inner sep=0,minimum size=3pt}]
    \draw[->] (0,0) -- (8,0) node[right] {time};
    \draw[thick] (0.5,0) -- +(0,0.2) node[above] {$(k-1)Q$};
    \draw[thick] (1.5,0) -- +(0,0.2) node[above] {$kQ$};
    \draw[thick] (3.3,0) -- +(0,0.2) node[above] {$(k+c)Q$};
    \draw[] (4.6,0) -- +(0,0.2) node[above] {$kQ+T$};
    \draw[] (6.6,0) -- +(0,0.2) node[above] {$(k+c+1)Q+T$};
    \draw[] (3.5,0) -- +(0,-0.2) node[below] {$t$};
    \draw[] (1.7,0) -- +(0,-0.2) node[below] {$t-D^*$};
    \draw[] (5.5,0) -- +(0,-0.2) node[below] {$t+D^*$};
  \end{tikzpicture}
   }
  \caption{Illustration of the proof of Theorem~\ref{thm:cf-2}. A thick tick indicts the start position of a submatrix $S[T,Q,k]$, and a thin tick indicts the time.} \label{fig:cf-2}
\end{figure}

The following lemma improves Lemma~\ref{lemma:cfp-2} for the binary collision model.
The proof is the same as that of Lemma~\ref{lemma:cfp-2} except that  Theorem~\ref{thm:cf-2} is applied instead of Theorem~\ref{thm:cf-3}.

\begin{lemma}\label{lemma:cfp}
For a network $\mc N$ with a binary collision profile and integers $T$ and $Q$ such that $T\geq D^*$ and $1\leq Q\leq T$, 
suppose $S$ is a periodic schedule with period $kQ$ such that $(S[T,Q,i], i=0,\ldots,k)$ is a closed path in $(\mc M_T, \mc E_{T,Q})$. Then $S$ is collision-free for $\mc N$.
\end{lemma}

The following theorem improves Theorem~\ref{thm:reg-3} for the binary collision model.
The proof is the same as that of Theorem~\ref{thm:reg-3} except that Lemma~\ref{lemma:cfp} is applied instead of Lemma~\ref{lemma:cfp-2}.

\begin{theorem}
	\label{thm:reg-2}
	For a network $\mc N$ with a binary collision profile and any integers $T$ and $Q$ such that $T\geq D^*$ and $1\leq Q\leq T$, it holds that 
		$\cR^{\mc N} \supset \mc R^{(\mc M_T,\mc E_{T,Q})}$.
\end{theorem}

\section{Algorithms for Calculating Scheduling Rate Region}
\label{sec:alg}

In this section, our focus is on developing algorithms for calculating the scheduling rate region of a network $\mc N = (\mc L, \mc I, D_{\mc L})$.
Based on the discussion in Sec.~\ref{sec:sg}, we understand that the rate region of
$\mc N$ is determined by cycles of the scheduling graph
$(\mc M_{T},\mc E_{T,Q})$ with a sufficiently large $T$. Therefore, a
straightforward approach to compute the rate region is to apply
the existing cycle enumerating algorithms on the scheduling graphs
(e.g., Johnson's algorithm~\cite{johnson1975finding}).  Though we may
choose the minimum value of $T$ to reduce the computation cost, the
straightforward approach in general incurs a high computation
cost.
Note that $|\mc M_{T}| =O(2^{|\mc L|T})$ and the number of
cycles in $(\mc M_{T},\mc E_{T,Q})$ can exceed
$2^{|\mc M_{T}|}$~\cite{johnson1975finding}. Consequently, cycle
enumerating algorithms designed for generic graphs tend to exhibit a steep increase in computation cost as a function of $|\mc L|T$.

In the existing scheduling researches, to avoid the high the
computational complexity of enumerating all the maximal independent
sets, algorithms have been developed to achieve a subset of the
scheduling rate
region~\cite{lin2006impact,chaporkar2008throughput}. Similarly, due to
the even greater computational challenge of enumerating all cycles in
a scheduling graph, it is also a reasonable approach to enumerate cycles
up to a specific length $k$, which can provide an approximation of the
scheduling rate region. Although the maximum cycle length
theoretically is $O(|\mc M_{T}|)$, both analytical and numerical
evidence suggests that a small value of $k$ can yield a good
approximation of the rate region.

Though an algorithm designed for generic graphs can be applied on a
scheduling graph to enumerate cycles up to a specific length
(e.g.,~\cite{liu2006new}), its running time may not be optimal since
it does not utilize the specific structure of the scheduling graph.
In this section, we present an approach specifically designed for the
step-$T$ scheduling graph $(\mc M_T, \mc E_{T})$. By leveraging a
dominance property of $(\mc M_T, \mc E_{T})$, we refine the
characterization of the scheduling rate region using only subgraphs of
$(\mc M_T, \mc E_{T})$. We derive algorithms that calculate the subset
of the scheduling rate region generated by the cycles of
$(\mc M_T, \mc E_{T})$ up to a specific length.  Numerical evaluations
demonstrate that our algorithms can achieve faster computation
compared to using generic cycle enumeration algorithms directly on
$(\mc M_T, \mc E_{T})$, particularly when dealing with larger
networks.  Moreover, the techniques developed here can be adopted in
the next section for maximizing a linear function of the rate vectors.

We denote $\max_{\pgeq}\mc A$ as the set of maximal elements in the
partially ordered set $(\mc A, \pgeq)$.  In other words,
$\max_{\pgeq}\mc A$ is the smallest subset $\mc B$ in $\mc A$ such
that any element of $\mc A$ is dominated by some elements in
$\mc B$.  A sequence of matrices $A = (A_0,A_1,\ldots)$, where
$A_i\in \{0,1\}^{|\mc L|\times t_i}$, is regarded as a matrix obtained
by juxtaposing $A_0, A_1, \ldots$.  Therefore, the relation $\pgeq$
and $\pleq$ defined on matrices can be applied to pairs of sequence
of matrices.  For two real numbers $a$ and $b$, we define
$a\land b$ as the minimum of $a$ and $b$.  For two matrices $A=(a_{ij})$ and
$B=(b_{ij})$ of the same size, we define
$A\land B = (a_{ij}\land b_{ij})$. When $A$ and $B$ are binary
matrices, $A\land B$ is the matrix resulting from the bitwise AND operation.

\subsection{Calculation of Scheduling Graphs}
\label{sec:sgcal}

Before delving into our approach to the rate region, let's discuss
the calculation of scheduling graphs. We first establish the equivalence between $(\mc M_{T},\mc E_{T,Q})$ and $\mc M_{T+Q}$. According to Definition~\ref{def:schedule_graph}, $\mc M_{T+Q}$ represents the subset of $|\mc L|\times (T+Q)$ binary matrices that correspond to the independent sets of $\mc N^{T+Q}$. On one hand, $\mc E_{T,Q}$ consists of pairs $(A,B) \in \mc M_T\times \mc M_T$ such that $(A[Q,Q,0],B) \in \mc M_{T+Q}$ and the last $T-Q$ columns of $A$ and the first $T-Q$ columns of $B$ are the same. On the other hand, for every $|\mc L|\times (T+Q)$ binary matrix $C\in \mc M_{T+Q}$, both $C[T,Q,0]$ and $C[T,Q,1]$ are elements of $\mc M_T$, resulting in $(C[T,Q,0],C[T,Q,1]) \in \mc E_{T,Q}$. Therefore, the calculation of $(\mc M_{T},\mc E_{T,Q})$ is equivalent to the calculation of $\mc M_{T+Q}$.

Denote $\mc M_{t}^* = \max_{\pgeq}\mc \mc M_{t}$ as the collection of the $|\mc L|\times t$ binary matrices that represent the \emph{maximal} independent sets of $\mc N^{t}$.
When dealing with a binary collision model, the Bron–Kerbosch algorithm and its refinements~\cite{bron1973algorithm,tomita2006worst,cazals2008note,eppstein2013listing} can be employed to enumerate $\mc M_t^*$.
In the case of a general collision model where $\mc N^{t}$ forms a hypergraph, the corresponding problem of finding maximal independent sets has been discussed in  \cite{beame1990parallel,dahlhaus1992efficient,bercea2014computing}. 
The worst-case complexity of the Bron–Kerbosch algorithm is
$O(3^{n/3})$, where $n$ is the number of vertices in the
network~\cite{tomita2006worst}. In our experience, the vertex pivoting technique \cite{tomita2006worst,cazals2008note} can greatly improve the running time of the Bron–Kerbosch
algorithm for $\mc N^{t}$. 

\begin{algorithm}[tb]
  \caption{A algorithm for calculating $(\mc M_T, \mc E_{T,Q})$ from $\mc M_{T+Q}^*$. The function SchedGraph recursively call ADDC to add vertices and edges to $(\mc M_T, \mc E_{T,Q})$.} \label{alg:sg}
  \begin{algorithmic}[1]
    \Function{SchedGraph}{}
      \State \textbf{Input:} $\mc M_{T+Q}^*$, $T$, $Q$
      \State \textbf{Output:} $(\mc M_T, \mc E_{T,Q})$
      \State $(\mc M_T,\mc E_{T,Q}) \leftarrow (\emptyset,\emptyset)$
      \For{each $C \in \mc M_{T+Q}^*$}
        \State ADDC$(C)$
      \EndFor
    \Procedure{ADDC}{}
      \State \textbf{Input:} $C$
      \State \textbf{Output:} update $(\mc M_T, \mc E_{T,Q})$
      \If{$C[T,Q,0]\in \mc M_T$, $C[T,Q,1]\in \mc M_T$ and $(C[T,Q,0],C[T,Q,1])\in \mc E_{T,Q}$} \label{alg:sg:10}
        \State return
      \EndIf
      \State $\mc M_T \leftarrow \mc M_T \cup \{C[T,Q,0], C[T,Q,1]\}$
      \State $\mc E_T \leftarrow \mc E_T \cup \{(C[T,Q,0],C[T,Q,1])\}$
      \If{$C$ is the all-zero matrix}
        \State return
      \EndIf
      \For{each $C'\precneqq C$}
        \State ADDC$(C')$
      \EndFor
    \EndProcedure      
    \EndFunction
  \end{algorithmic}
\end{algorithm}

For the straightforward approach to calculating the scheduling rate region, we require $\mc M_{T+Q}$ rather than just $\mc M_{T+Q}^*$. Given $\mc M_{T+Q}^*$, we can generate $(\mc M_{T},\mc E_{T,Q})$ as follows:
For each $C \in \mc M_{T+Q}^*$, and recursively for each $C' \pleq C$, we add $C'[T,Q,0]$ and $C'[T,Q,1]$ to $\mc M_T$, and we add $(C'[T,Q,0],C'[T,Q,1])$ to $\mc E_{T,Q}$.
In Algorithm~\ref{alg:sg}, we provide pseudocode for calculating $(\mc M_{T},\mc E_{T,Q})$ from $\mc M_{T+Q}^*$. Line~\ref{alg:sg:10} checks whether $C$ has already been added to $(\mc M_{T},\mc E_{T,Q})$ before.

In our approach to scheduling the rate region (to be elaborated in this section), we will use two subgraphs of the step-$T$ scheduling graph $(\mc M_{T},\mc E_{T})$. Therefore, we do not require the computation of the entire scheduling graph based on $\mc M_{2T}^*$.
The first subgraph of interest corresponds to $\mc M_{2T}^*$. If we consider $(A,B) \in \mc E_{T}$ as an $|\mc L|\times 2T$ binary matrix, then $(A,B) \in \mc M_{2T}$. Thus, we can express $\mc E_{T} = \mc M_{2T}$.
Let 
$\mc E^* = \max_{\pgeq}\mc E_{T} = \mc M_{2T}^*$, and let
\begin{IEEEeqnarray*}{rCl}
  \mc M^*_L & = & \{B: (B,B')\in \mc E^* \text{ for certain } B'\}, \\
  \mc M^*_R & = & \{B': (B,B') \in \mc E^* \text{ for certain } B\}. 
\end{IEEEeqnarray*}
As $\mc E^* \subset \mc M^*_L\times \mc M^*_R$, $\mc E^*$ can be represented using an adjacency matrix with rows and columns indexed by elements in $\mc M^*_L$ and $\mc M^*_R$, respectively. We see that $(\mc M^*_L,\mc M^*_R,\mc E^*)$ is just another representation of  $\mc M_{2T}^*$ and serves as a subgraph of $(\mc M_{T},\mc E_{T})$. The second subgraph of interest can be induced by $(\mc M^*_L,\mc M^*_R,\mc E^*)$, which will be discussed later in this section. 

\subsection{Dominance Property}
\label{sec:dom}

According to the definition of collision, if a schedule is collision-free, then the schedule obtained by inactivating some entries is also collision-free. In other words, if $A \in \mc M_T$, then any $A'\pleq A$ is also in $\mc M_T$. 
The similar property applies to edges and paths in $(\mc M_T,\mc E_T)$. 
For two sequences $A = (A_0,A_1,\ldots)$ and $B = (B_0,B_1,\ldots)$ of
the same length (which can be unbounded) with
$A_i,B_i\in \{0,1\}^{|\mc L|\times T}$, we say $A$ \emph{dominates}
$B$ if $A \pgeq B$. 

\begin{lemma}[Basic dominance property for $(\mc M_T,\mc E_T)$]\label{lemma:dom}
  For any $k\geq 0$, if $A =(A_0,A_1,\ldots,A_k)$ is a path in $(\mc M_T,\mc E_T)$, then  any $B =(B_0,B_1,\ldots,B_k)$ with $B\pleq A$ is a path in $(\mc M_T,\mc E_T)$.
\end{lemma}
\begin{IEEEproof}
  For any edge $(A',A'') \in \mc E_T$ and any $(B',B'')$, if $(B',B'') \pleq (A',A'')$, then $(B',B'')$ is also an edge of $(\mc M_T,\mc E_T)$. The lemma can then be proved by checking $(B_i,B_{i+1})\pleq (A_i,A_{i+1})$ for $i=0,1,\ldots,k-1$. 
\end{IEEEproof}

We now define some notations for presenting a main dominance property. Denote $\mc P_k$ as the set of length-$k$ paths and $\mc C_k$ as the set of length-$k$ cycles in $(\mc M_{T}, \mc E_{T})$.
Let
\begin{equation*}
  \mc P_k^* = \max_{\pgeq} \mc P_k \quad \text{and} \quad \mc C_k^* = \max_{\pgeq} \mc C_k.
\end{equation*}
In other words, the elements in $\mc P_k^*$ and $\mc C_k^*$ are  the {maximal} paths and cycles, respectively, with respect to the partial order $\pgeq$.
Note that $\mc C_k\subset \mc P_k$, but $\mc C_k^*$ is not necessarily a subset of $\mc P_k^*$.
We are going to show that maximal paths and cycles are sufficient for characterizing the scheduling rate region.

Let
\begin{equation} \label{eq:rk}
  \mc R_k=\conv\{R_C: C\in \cup_{i=1}^k \mc C_i\},
\end{equation}
where $R_C$ is defined in Definition~\ref{def:rcp} with $Q=T$. It is worth noting that as  $k$ becomes sufficiently large,
 $\mc R_k$ becomes $\mc R^{(\mc M_T,\mc E_{T})}$, which is equal to $\mc R^{\mc N}$ when $T\geq 2D^*$ (or $D^*$ for binary collision).
Therefore, when $T\geq 2D^*$ (or $D^*$ for binary collision) and $k$ is sufficiently large, $\mc R_k$ is the scheduling rate region.
The \emph{lower shadow} of a set $\mc A \subset (\mathbb R^+)^{m\times n}$, denoted as $\dom \mc A$, refers to the collection of all $B\in (\mathbb R^+)^{m\times n}$ that are dominated by some elements in $\mc A$ \cite{bollobas1998modern, schrijver2003combinatorial}. The following lemma states that the lower shadow of $\mc R_k$ is equal to the set $\mc R_k$ itself.

\begin{lemma}\label{lemma:dom}
  For each $k\geq 1$,
  $\dom\mc R_k = \mc R_k$.
\end{lemma}
\begin{IEEEproof}
  As $\mc R_k \subset \dom\mc R_k$, we show $\dom\mc R_k\subset \mc R_k$. For $R \in \dom\mc R_k$, there exists $R' \in \mc R_k$ such that $R'\pgeq R$.
Let $\mc C^* = \cup_{i=1}^k \mc C_i$.
We can write $R' = \sum_{C'\in \mc C^*} \alpha_{C'} R_{C'}$, where $\alpha_{C'}\geq 0$ and $\sum_{C'\in \mc C^*} \alpha_{C'} = 1$. Assume $R(l) < R'(l)$ for a certain $l\in \mc L$. For each $C'\in \mc C^*$, construct a cycle $C$ by setting the entries of the matrices in $C'$ indexed by $l$ to zero.
Let $\alpha_l = \frac{R(l)}{R'(l)}$ and
  \begin{equation*}
    R'' = \alpha_l \sum_{C'\in \mc C^*} \alpha_{C'} R_{C'} + (1-\alpha_l) \sum_{C'\in \mc C^*} \alpha_{C'} R_{C}.
  \end{equation*}
  We have $R'' \in \mc R_k$ as $C'\in \mc C^*$, $R''(l) = R(l)$ and $R''(l') = R'(l')$ for $l'\neq l$.  By repeating the similar procedure for all the other links $l'$ with $R'(l') > R(l')$, we can convert $R''$ to $R$ and hence prove $R \in \mc R_k$.
\end{IEEEproof}

For a path $P=(A_0,A_1,\ldots,A_k)$ of length $k$, we define $\cl(P)$ as the closed path generated by $$(A_0\land A_k, A_1,\ldots,A_{k-1},A_0\land A_k).$$
Hence, the rate vector $R_{\cl(P)}$ is well-defined, following Definition~\ref{def:rcp} with $Q=T$.
The following theorem shows that $\mc R_k$ can be determined by $\mc P_i^*$ and $\mc C_i^*$, $i=1,\ldots,k$.

\begin{theorem}\label{thm:maxcyc}
  For a scheduling graph $(\mc M_{T}, \mc E_{T})$ and any integer $k\geq 1$,
  \begin{IEEEeqnarray*}{rCl}
    \mc R_k & = & \dom\ \conv\{R_{C}:C\in \mc \cup_{i=1}^k\mc C^*_i\}  \\
     & = & \dom\ \conv\{R_{\mathrm{cl}(P)}:P\in \mc \cup_{i=1}^k\mc P^*_i\}.
  \end{IEEEeqnarray*}
\end{theorem}
\begin{IEEEproof}
  To simplify the notation, we use in the proof 
  \begin{IEEEeqnarray*}{rCl}
    \mc A^* & = & \dom\ \conv\{R_{C}:C\in \mc \cup_{i=1}^k\mc C^*_i\}, \\
    \mc B^* & = & \dom\ \conv\{R_{\mathrm{cl}(P)}:P\in \mc \cup_{i=1}^k\mc P^*_i\}.
  \end{IEEEeqnarray*}
  By definition, $\mc R_k \subset \mc A^*$.
  
  For $A=(A_0,\ldots,A_{k-1},A_0)\in \mc C_k^*$, there exists $P=(B_0,\ldots,B_k)\in \mc P_k^*$ such that $B\pgeq A$. As $B_0\pgeq A_0$ and $B_k\pgeq A_0$, $B_0\land B_k\pgeq A_0$. So, $\cl(P) \pgeq A$, and hence $R_A \pleq \conv\{R_{\mathrm{cl}(P)}:P\in \mc \mc P^*_k\}$. Therefore, $\mc A^*\subset \mc B^*$. 

  Last, for any $P=(B_0,\ldots,B_k)\in \mc P_k^*$, if $\cl(P)$ is a cycle, then $R_{\cl(P)} \in \mc R_k$. When $k=1$, $\cl(P)$ must be a cycle. When $k> 1$, 
if $\cl(P)$ is not a cycle, it can be decomposed into multiple cycles, each of which is of length strictly less than $k$. Hence, $R_{\cl(P)} \in \mc R_{k-1} \subset \mc R_k$. The proof is completed by $\mc B^* \subset \dom(\mc R_k) = \mc R_k$ where the equality follows from Lemma~\ref{lemma:dom}. 
\end{IEEEproof}

\subsection{An Incremental Approach for Rate Region}
\label{sec:amp}

We are motivated to study the calculation of $\mc R_k$ due to  it's relation to the rate region. 
Based on Theorem~\ref{thm:maxcyc}, we will derive an approach to calculate $\mc R_k$ using $(\mc M^*_L,\mc M^*_R,\mc E^*)$.

As $\mc P^*_1 = \mc E^*$, according to Theorem~\ref{thm:maxcyc}, we have
\begin{equation}\label{eq:r1}
  \mc R_1 =\dom\ \conv \{R_{\cl(P)}: P\in \mc E^*\}.
\end{equation}
We can use an incremental approach to calculate $\mc R_k$ for $k\geq 2$. 
Denote $\cv 1$ as a column vector with all entries equal to $1$, where the length is known from the context. 
For $A\in \mc M_L^*$ and $B\in \mc M_R^*$, let
\begin{equation}\label{eq:r2ab}
  \mc W_2(A,B) = \max_{\pgeq} \{ (B_1\land A_2)\cv 1: (A,B_1), (A_2,B)\in \mc E^*\}. 
\end{equation}
For $k\geq 3$, $A\in \mc M_L^*$ and $B\in \mc M_R^*$, let
\begin{equation}\label{eq:rkab}
  \mc W_k(A,B) = \max_{\pgeq} \bigcup_{B'\in \mc M_R^*} \left(\mc W_{k-1}(A,B') + 
    \{(B'\land A')\cv 1: (A',B)\in \mc E^* \} \right),
\end{equation}
where the addition of two sets $\mc A + \mc B$ is defined as $\{\cv a + \cv b: \cv a \in \mc A, \cv b\in \mc B\}$. 
The next theorem justifies the use of $\mc W_k(A,B)$ to characterize the rate region $\mc R_k$. The proof of this theorem will be provided at the end of this subsection.

\begin{theorem}\label{thm:uu}
  Consider the step-$T$ scheduling graph $(\mc M_T, \mc E_T)$.   For $k\geq 1$, we have 
  ${\mc R}_k = \dom\ \conv\ \mc R_1^* \cup \mc R_2^* \cup \cdots \cup \mc R_k^*$, where
  \begin{equation*}
    \mc R_1^* = \frac{1}{T}\max_{\pgeq} \{(A\land B)\cv 1: (A,B)\in \mc E^*\},
  \end{equation*}
  and for $i\geq 2$
  \begin{equation}
    \label{eq:rkstar}
    \mc R_i^* = \frac{1}{iT} \max_{\pgeq} \bigcup_{A\in \mc M_L^*,B\in \mc M_R^*} \left({\mc W}_i(A,B)+\{(A\land B)\cv 1\}\right).
  \end{equation}
\end{theorem}

Before presenting the proof of Theorem~\ref{thm:uu}, we discuss the algorithms for calculating $\mc W_k(A,B)$ and $\mc R_k^*$.
According to Theorem~\ref{thm:uu}, using $\mc R_1^*, \ldots, \mc R_k^*$, the vertex representation of the
convex polytope $\mc R_k$ can be derived, which can then be converted to the half-space representation~\cite{grunbaum2003convex}.

\subsubsection{Algorithm for Rate Region Calculation}

Algorithm~\ref{alg:rab} provides the pseudocode for calculating
$\mc W_k(A,B)$ incrementally using \eqref{eq:rkab}, and Algorithm~\ref{alg:region} provides the pseudocode for calculating $\mc R_k^*$ using the formula in Theorem~\ref{thm:uu}. These algorithms assume that $(\mc M^*_L,\mc M^*_R,\mc E^*)$ has already been calculated (as mentioned in Sec.\ref{sec:sgcal}). The explanations for these two algorithms are as follows. To simplify the notation, we write
\begin{equation*}
  \mc W_{k} \triangleq (\mc W_{k}(A,B),A\in \mc M_L^*, B\in \mc M_R^*).
\end{equation*}

In Algorithm~\ref{alg:rab}, three functions are provided: W2AB, WAB and MAXADD. 
The W2AB function calculates $\mc W_2$ using \eqref{eq:r2ab}, and 
the WAB function calculates $\mc W_{k}$ from $\mc W_{k-1}$ using \eqref{eq:rkab}. Line~\ref{alg:1:7} and Line~\ref{alg:1:11} call the
function MAXADD to add a vector $\cv u$ to a set $\mc S$ that
contains all the existing maximal vectors.  If $\cv u$ is
dominated by some vector in $\mc S$, $\mc S$ remains unchanged. If some
vectors in $\mc S$ are dominated by $\cv u$, they are deleted from
$\mc S$, followed by adding $\cv u$ to $\mc S$.

The computation cost of WAB depends on the size of $\mc W_k(A,B)$. Let
\begin{equation*}
  W_k = \max_{A\in \mc M_L^*, B\in \mc M_R^*}|\mc W_k(A,B)|.
\end{equation*}
According to the definitions in \eqref{eq:r2ab} and \eqref{eq:rkab}, we have $W_2 \leq |\mc M_L^*||\mc M_R^*|$, and for $k>2$, $W_k\leq |\mc M_L^*||\mc M_R^*| W_{k-1}$. 
The computation cost from $W_{k-1}$ to $W_{k}$ using WAB can be estimated as
\begin{equation*}
  O(|\mc M_L^*|^2|\mc M_R^*|^2 W_{k-1}(W_k+|\mc L|T))
\end{equation*}
integer and logical operations.

\begin{algorithm}[tb]
  \caption{The pseudocode for calculating $\mc W_{k}$ includes two functions: W2AB, which calculates $\mc W_{2}$, and WAB, which calculates $\mc W_{k}$ from $\mc W_{k-1}$ for any $k>2$. Additionally, there is a helper function called MAXADD, which is called by both W2AB and WAB to add an element to a set and output the maximal subset.} \label{alg:rab}
  \begin{algorithmic}[1]
    \Function{W2AB}{}
      \State \textbf{Input:} $\mc M_L^*,\mc M_R^*,\mc E^*$
      \State \textbf{Output:} $\mc W_{2}$
      \For{each $A\in \mc M_L^*$ and $B\in \mc M_R^*$}
        \State $\mc W_2(A, B) \leftarrow \emptyset$
        \For{each $B_1$ s.t. $(A,B_1)\in \mc E^*$ and $A_2$ s.t. $(A_2,B)\in \mc E^*$ }
          \State $\mc W_2(A,B) \leftarrow$ MAXADD($\mc W_2(A,B)$, $(B_1\land A_2)\cv 1$) \label{alg:1:7}
        \EndFor
      \EndFor
      \State return $\mc W_{2}$
    \EndFunction
    \Function{WAB}{}
      \State \textbf{Input:} $\mc W_{k-1}, \mc M_L^*,\mc M_R^*,\mc E^*$
      \State \textbf{Output:} $\mc W_{k}$
      \For{each $A\in \mc M_L^*$ and $B\in \mc M_R^*$}
        \State $\mc W_k(A, B) \leftarrow \emptyset$
        \For{each $B'\in \mc M_R^*$ and $A'\in \mc M_L^*$ s.t. $(A',B)\in \mc E^*$}
            \For{each $\cv r\in \mc W_{k-1}(A,B')$}
              \State $\mc W_k(A,B) \leftarrow$ MAXADD($\mc W_k(A,B)$, $\cv r +(B'\land A')\cv 1$) \label{alg:1:11}
            \EndFor
        \EndFor
      \EndFor
      \State return $\mc W_{k}$
    \EndFunction
    \Function{MAXADD}{}
      \State \textbf{Input:} a set $\mc S$ of maximal vectors, a vector $\cv u$
      \State \textbf{Output:} $\max_{\pgeq} \mc S \cup \{\cv u\}$
      \For{each $\cv r \in \mc S$}
        \If{$\cv u\pleq \cv r$}
          \State return
        \EndIf
        \If{$\cv r\pleq \cv u$}
          \State $\mc S \leftarrow \mc S\setminus \{\cv r\}$
        \EndIf
      \EndFor
      \State $\mc S \leftarrow \mc S \cup \{\cv u\}$
      \State return $\mc S$
    \EndFunction
  \end{algorithmic}
 \end{algorithm}

In Algorithm~\ref{alg:region}, two functions are provided: RateRegion and CALRR.
The RateRegion function takes an integer $k_{\max}$ as input and calculates $\mc R_k^*$ for $k=1,2,\ldots,k_{\max}$ as output. The function calls W2AB and WAB to obtain $\mc W_k$ for $k=2,\dots,k_{\max}$ and then calculates $\mc R_k^*$ by calling CALRR on $\mc W_k$, using the formula provided in Theorem~\ref{thm:uu}.

The value of $k_{\max}$ has an impact on both the computation cost and the subset of the rate region obtained. 
The computation cost of CALRR for $\mc W_k$ is $O(|\mc M_L^*|^2|\mc M_R^*|^2W_k^2)$ integer and logic operations. Assuming $W_k\geq |\mc L|T$, the overall computation cost of RateRegion is
$O(k_{\max}|\mc M_L^*|^2|\mc M_R^*|^2W_{k_{\max}}^2)$. In the worst case, $W_k$ grows exponentially with $k$, and hence using a larger value of $k_{\max}$ may significantly increase the computation cost. However, using a larger value of $k_{\max}$ allows for obtaining a larger subset of the rate region. %

Furthermore, it is worth noting that as $k$ increases, $\mc R_k$ converges to $\mc R^{(\mc M_T,\mc E_{T})}$. In other words, for sufficiently large values of $k$, evaluating $\mc R_k^*$ does not result in an increase in the
polytope $\mc R_k$. Let $k^*$ denote the smallest
value of $k$ such that $\mc R_k=\mc R^{(\mc M_T,\mc E_{T})}$.  An
upper bound on $k^*$ is the largest cycle length in
$(\mc M_T,\mc E_{T})$, which is not greater than $|\mc M_T|$.
However, it is important to highlight that $k^*$ can be much smaller than the largest cycle length in $(\mc M_T,\mc E_{T})$. This will be demonstrated later, showing that the convergence can occur much earlier than expected based on the largest cycle length.

\begin{algorithm}[tb]
  \caption{The pseudocode for calculating $\mc R_k^*$ for $k=1, 2,\ldots,k_{\max}$ includes two functions: RateRegion and CALRR. The RateRegion function calculates $\mc R_k^*$ for $k=1, 2,\ldots,k_{\max}$, while the CALRR function is called by RateRegion to calculate $\mc R_i^*$ using $\mc W_i$.}\label{alg:region}
  \begin{algorithmic}[1]
    \Function{RateRegion}{}
      \State \textbf{Input:} $(\mc M^*_L, \mc M^*_R,\mc E^*)$, integer $k_{\max}$
      \State \textbf{Output:} $(\mc R_k^*,k=2,\ldots,k_{\max})$
      \State $\mc R_1^* \leftarrow \emptyset$
      \For{each $A\in \mc M_L^*$ and $B\in \mc M_R^*$ s.t. $(A,B)\in \mc E^*$}
        \State ${\mc R}_1^* \leftarrow$ MAXADD(${\mc R}_1^*$, $\frac{1}{T}(B\land A)\cv 1$)
      \EndFor
      \State $\mc W_2 \leftarrow $ W2AB$(\mc M^*_L, \mc M^*_R,\mc E^*)$
      \State $\mc R_2^* \leftarrow \frac{1}{2T}$ CALRR$(\mc W_2,\mc M_L^*,\mc M_R^*)$
      \For{$k$ from $3$ to $k_{\max}$}
        \State $\mc W_k \leftarrow $ WAB$(\mc W_{k-1},\mc M^*_L, \mc M^*_R, \mc E^*)$
        \State $\mc R_k^* \leftarrow \frac{1}{Tk}$ CALRR$(\mc W_k,\mc M_L^*,\mc M_R^*)$
      \EndFor
      \State return $(\mc R_k^*,k=1,2,\ldots,k_{\max})$
    \EndFunction
    \Function{CALRR}{}
      \State \textbf{Input:} $\mc W_k$, $\mc M_L^*$, $\mc M_R^*$
      \State \textbf{Output:} $\tilde{\mc R}_k$
      \State $\tilde{\mc R}_k \leftarrow \emptyset$      
      \For{each $A\in \mc M_L^*$ and $B\in \mc M_R^*$}
        \For{each $\cv r \in \mc W_k(A,B)$ }
          \State $\tilde{\mc R}_k \leftarrow$ MAXADD($\tilde{\mc R}_k$, $\cv r +(B\land A)\cv 1$)
        \EndFor
      \EndFor
      \State return $\tilde{\mc R}_k$
    \EndFunction
  \end{algorithmic}
\end{algorithm}

\subsubsection{Proof of Theorem~\ref{thm:uu}}

The following lemma gives an incremental approach to enumerate a \emph{superset} of $\mc P_k^*$ incrementally.

\begin{lemma}\label{le:pp}
  For any $k\geq 2$, $\mc P_k^*$ is a subset of
\begin{IEEEeqnarray*}{rCl}
  \{ (B_0,\ldots,B_{k-2}, B_{k-1}\land B_{k-1}' ,B_k): (B_{k-1}',B_{k}) & \in & \mc E^*, \\
  (B_0,\ldots,B_{k-1}) & \in & \mc P_{k-1}^*\}.
\end{IEEEeqnarray*}
\end{lemma}
\begin{IEEEproof}
For any $(A_0,A_1,\ldots,A_k) \in \mc P_k^*$, there exist $(B_0,B_1,\ldots,B_{k-1})\in \mc P_{k-1}^*$ and $(B_{k-1}',B_{k})\in \mc E^*$ such that
\begin{IEEEeqnarray*}{rCl}
  (B_0,B_1,\ldots,B_{k-1}) & \pgeq & (A_0,A_1,\ldots,A_{k-1}), \\
  (B_{k-1}',B_{k}) & \pgeq & (A_{k-1},A_k).
\end{IEEEeqnarray*}
We see $(B_0,\ldots,B_{k-2},B_{k-1}\land B_{k-1}'
,B_k)$ is a path of length $k$ and dominates $(A_0,A_1,\ldots,A_k)$. As the latter is maximal, we further have $B_i=A_i$ for $i=0,1,\ldots,k-2,k$ and $B_{k-1}\land B_{k-1}' = A_{k-1}$.  
\end{IEEEproof}

Let $\mc H_1 = \mc E^*$. 
For $k\geq 2$, let
\begin{equation*}
  \mc H_k  =  \{(A_1,B_1\land A_2,\ldots,B_{k-1}\land A_k,B_k): (A_i,B_i)\in \mc E^*, i=1,\ldots,k\}.
\end{equation*}

\begin{lemma}\label{lemma:css}
  For $k\geq 1$, each element of $\mc H_k$ is a path in $(\mc M_T, \mc E_T)$, and $\mc P_k^* \subset \mc H_k$. Moreover, for $k\geq 1$,
  \begin{equation*}
    \mc R_k = \dom \ \conv \{R_{\mathrm{cl}(P)}:P\in \mc \cup_{i=1}^k\mc H_i\}.
  \end{equation*}
\end{lemma}
\begin{IEEEproof}
  If $(A_0,A_1,\ldots,A_k)\in \mc H_k$, then $(A_{i-1},A_i)\in \mc E_T$, $i=1,\ldots,k$. Hence,
  each element of $\mc H_k$ is a path in $(\mc M_T, \mc E_T)$. The lemma holds directly when $k=1$.
  Now we consider $k\geq 2$. 
By Lemma~\ref{le:pp}, $\mc P_2^* \subset \mc H_2$. As 
\begin{IEEEeqnarray*}{rCl}
  \mc H_k & = & \{ (B_0,\ldots,B_{k-2}, B_{k-1}\land B_{k-1}' ,B_k): \\
  & & (B_{k-1}',B_{k}) \in \mc E^*,  (B_0,\ldots,B_{k-1}) \in \mc H_{k-1}\},
\end{IEEEeqnarray*}
by induction, $\mc P_k^* \subset \mc H_k$.
By Theorem~\ref{thm:maxcyc},
\begin{IEEEeqnarray*}{rCl}
  {\mc R}_k & = & \dom\ \conv\{R_{\mathrm{cl}(P)}:P\in \mc \cup_{i=1}^k\mc P^*_i\} \\
  & \subset & \dom \ \conv \{R_{\mathrm{cl}(P)}:P\in \mc \cup_{i=1}^k\mc H_i\}.
\end{IEEEeqnarray*}
For any $P\in \mc \cup_{i=1}^k\mc H_i$, $\mathrm{cl}(P)$ is a cycle of length at most $k$, and hence $R_{\mathrm{cl}(P)} \in \dom\ \conv\{R_{C}:C\in \mc \cup_{i=1}^k\mc C^*_i\} = \mc R_k$. The proof is complete.
\end{IEEEproof}

Now we are ready to prove Theorem~\ref{thm:uu}.

\begin{IEEEproof}[Proof Theorem~\ref{thm:uu}]
  The case for $k=1$ is proved by \eqref{eq:r1}.
  We first prove by induction that for $k\geq 2$, for any $(A_1,B_1), \ldots, (A_k,B_k) \in \mc E^*$, 
  \begin{equation}\label{eq:cis}
    \sum_{i=1}^{k-1} (B_i\land A_{i+1})\cv 1 \in \dom\ \mc W_k(A_1,B_{k}).
  \end{equation}
  First, \eqref{eq:cis} holds when $k=2$ by the definition of $\mc W_2(A,B)$ in \eqref{eq:r2ab}.  For $k\geq 3$, suppose \eqref{eq:cis} holds for $k-1$. Then for any $(A_1,B_1), \ldots, (A_k,B_k) \in \mc E^*$, 
  \begin{IEEEeqnarray*}{rCl}
    \sum_{i=1}^{k-1} (B_i\land A_{i+1})\cv 1
    & = &  \sum_{i=1}^{k-2} (B_i\land A_{i+1})\cv 1 + (B_{k-1}\land A_k)\cv 1 \\
    & \in & \dom\ \mc R_{k-1}(A_1,B_{k-1}) + (B_{k-1}\land A_k)\cv 1 \\
    & \in & \dom\ \mc W_{k}(A_1,B_k).
  \end{IEEEeqnarray*}

Let
\begin{equation*}
  \tilde{\mc R}_k =  \{R_{\cl(P)}: P \in \mc H_k\}.
\end{equation*}
By Lemma~\ref{lemma:css}, ${\mc R}_k = \dom\ \conv (\cup_{i=1}^k \tilde{\mc R}_i)$ for $k\geq 1$. As $\mc R_1 = \tilde{\mc R}_1$, the theorem is proved if we can show that for $k\geq 2$,
  \begin{equation*}
    \dom\ \tilde{\mc R}_k  =\dom\ \mc R_k^* = 
    \frac{1}{Tk} \dom \bigcup_{A\in \mc M_L^*,B\in \mc M_R^*} \left({\mc R}_k(A,B)+(A\land B)\cv 1\right).
  \end{equation*}

For each $R \in \tilde{\mc R}_k$, there exists $(A_1,B_1), \ldots, (A_k,B_k)\in \mc E^*$ such that 
  \begin{equation*}
    Tk R  = (A_1\land B_k) \cv 1 + \sum_{i=1}^{k-1} (B_i\land A_{i+1})\cv 1.
  \end{equation*}
  By \eqref{eq:cis}, $\sum_{i=1}^{k-1} (B_i\land A_{i+1})\cv 1 \in \dom \mc W_k(A_1,B_{k})$, and hence $R\in \frac{1}{Tk} \dom (\mc W_k(A_1,B_{k}) + (A_1\land B_k) \cv 1) = \dom\ \mc R_k^*$.

  Fix $A\in \mc M_L^*$ and $B\in \mc M_R^*$. For $R\in {\mc R}_k(A,B)$, there exist $A_2\ldots,A_k \in \mc M_L^*$, $B_1,\ldots,B_{k-1} \in \mc M_R^*$ such that $(A,B_1), (A_2,B_2),\ldots, (A_{k-1},B_{k-1}), (A_k,B) \in \mc E^*$ and $R = \sum_{i=1}^{k-1} (B_i\land A_{i+1})\cv 1$. Hence, $\frac{1}{Tk}(R+ (A\land B)\cv 1) = \frac{1}{Tk}\left(\sum_{i=1}^{k-1} (B_i\land A_{i+1})\cv 1 + (A\land B)\cv 1\right) \in \tilde{\mc R}_k$. 
\end{IEEEproof}

\subsection{Reduced Scheduling Graph}
\label{sec:aep}

Now, we provide an alternative representation of the dominance property. By utilizing $(\mc M_L^*,\mc M_R^*, \mc E^*)$, 
we will derive a subgraph of $(\mc M_T, \mc E_T)$ and demonstrate that this subgraph can fulfill the same role as $(\mc M_T, \mc E_T)$ in characterizing the scheduling rate region. 
Let 
\begin{IEEEeqnarray*}{rCl}
  \mc V & = & \{ B\land B': B\in \mc M_R^*, B' \in \mc M_L^*\},\\
  \mc F & = & \{(B_1\land A_2,B_2\land A_3): B_1\in \mc M_R^*, (A_2,B_2)\in \mc E^*, A_3\in \mc M_L^*\}.
\end{IEEEeqnarray*}
The pair $(\mc V, \mc F)$ forms a directed graph that serves as a subgraph of $(\mc M_T, \mc E_T)$.

In Table~\ref{tab:2_hop}, we evaluate the sizes of  $\mc M_1, \mc E_1, \mc V, \mc F$ for the line network $\mc N_{L,2}^{\text{line}}$.
The following example illustrates a network with $|\mc M_T|$ and $|\mc E_T|$ of exponential functions of the number of links, while $\mc V$ and $\mc F$ have a constant size. 

\begin{table}[tb]
  \centering
  \caption{Evaluations of sizes of $(\mc M_1, \mc E_1)$ and $(\mc V, \mc F)$ for the line network with the $2$-hop collision model $\mc N_{L,2}^{\text{line}}$.} \label{tab:2_hop}
	\begin{tabular}{lrrrrrrrrrrrrr}
		\toprule $L$
		& $4$ & $5$& $6$& $7$ &$8$& $9$&$10$&$11$
		\\ \midrule
		$|\mc M_1|$ & $ 9 $ & $15  $ & $ 25  $ & $ 40 $ &  $ 64 $  &  $ 104 $  &  $169  $ & $ 273 $ 
		\\
		$|\mc E_1|$ &  $ 49 $ & $ 121 $ & $ 304 $ &  $ 676 $  & $ 1480 $   & $ 3481 $ & $ 8245 $ & $ 18769 $
		\\
		$|\mc V|$ & $ 9 $ & $ 9 $ & $ 16  $ & $ 30 $ &  $ 49 $  &  $ 72 $  &  $ 100 $ & $ 156 $  \\
          $|\mc F|$& $49  $ & $ 49 $ & $ 120  $ & $324  $ &  $ 800 $  &  $1681  $  &  $ 3074 $ & $ 6241 $  \\
		\bottomrule
	\end{tabular}
\end{table}

\begin{example}[Single collision network]\label{ex:1link}
  Consider a network of $L$ links with the link set $\mc L=\{l_1,\ldots,l_L\}$ and a binary collision model, where $\mc I(l_1) = \{l_2\}$ and $\mc I(l_i) = \emptyset$ for $i>1$.
  The delay matrix $D_{\mc L}$ has $D_{\mc L}(l_1,l_2) = 1$. We denote this network as $\mc N_{L}^{1\text{-c}}$, which has the character $D^* = 1$.
  The scheduling graph $(\mc M_1, \mc E_1)$ of this network has $\mc M_1 = \{0,1\}^{L}$. For $A, B\in \mc M_1$, $(A,B)\in \mc E_1$ if either i) $A(l_1)=0$ or ii) $A(l_1)=1$ and $B(l_2)=0$. Therefore, $|\mc M_1|=2^{L}$ and $|\mc E_1| = 2^{2L-1}+2^{2L-2}$, which increases exponentially with $L$.
  The reduced representation of $(\mc M_1, \mc E_1)$ has
  \begin{equation*}
    \mc E^* = \left\{(\cv 1,\cv v_2), (\cv v_1,\cv 1) \right\},
  \end{equation*}
  where $\cv 1$ is the all-$1$ vector of $L$ entries, $\cv v_1$ and $\cv v_2$ are obtained from $\cv 1$ by setting the first and second entry to $0$, respectively. Further $\mc M_L^* = \{\cv v_1, \cv 1\}$ and $\mc M_R^*= \{\cv v_2,\cv 1\}$. 
  Let's calculate $(\mc V, \mc F)$ for the single collision network $\mc N_{L}^{1\text{-c}}$. First, $\mc V = \{\cv 1, \cv v_1,\cv v_2, \cv v_1\land \cv v_2\}$. Then, $\mc F = \{(\cv v_1,\cv v_1), (\cv v_1,\cv 1), (\cv 1,\cv v_1\land \cv v_2), (\cv 1, \cv v_2), (\cv v_1\land \cv v_2,\cv v_1), (\cv v_1\land \cv v_2,\cv 1), (\cv v_2, \cv v_1\land \cv v_2), (\cv v_2, \cv v_2) \}$. We see that though the size of the scheduling graph $(\mc M_1, \mc E_1)$ of $\mc N_{L}^{1\text{-c}}$ is exponential in $L$, $(\mc V, \mc F)$ has a constant size. 
\end{example}

The graph $(\mc V, \mc F)$, called the \emph{reduced scheduling graph}, captures the essential connections and relationships from $(\mc M_T, \mc E_T)$. It is worth noting that a cycle in $(\mc V, \mc F)$ is also a cycle in $(\mc M_T, \mc E_T)$. However, it is not necessarily true that a cycle in $(\mc M_T, \mc E_T)$ is dominated by a cycle in $(\mc V, \mc F)$.
The following theorem demonstrates the possibility of characterizing the scheduling rate region using cycles in $(\mc V, \mc F)$.

\begin{theorem}\label{thm:k}
For a scheduling graph $(\mc M_T, \mc E_T)$, for $k\geq 1$,
\begin{equation*}
  \mc R_k = \dom\ \conv \{R_C: C \text{ is a length-$i$ cycle in } (\mc V,\mc F), i\leq k\}.
\end{equation*}
Therefore, $\mc R_k = \mc R^{(\mc M_T, \mc E_T)}$ when $k$ is the largest cycle length in $(\mc V, \mc F)$.
\end{theorem}
\begin{IEEEproof}
  To simplify the notation, let
  \begin{equation*}
    \mc A_k = \dom\ \conv \{R_C: C \text{ is a length-$i$ cycle in } (\mc V,\mc F), i\leq k\}.
  \end{equation*}
Consider a length-$k$ cycle $V = (V_0,V_1,\ldots,V_{k-1},V_k=V_0)$ of $(\mc V,\mc F)$.
There exist $(A_i,B_i) \in \mc E^*$ such that $(A_i,B_i)\pgeq (V_{i-1},V_i)$ for $i=1,\ldots,k$. As $P \triangleq (A_1,B_1\land A_2,\ldots,B_{k-1}\land A_k,B_k) \pgeq V$ and $P\in \mc H_k$, we have $R_V \pleq R_{\cl(P)} \in \mc R_k$, where the inclusion follows from Lemma~\ref{lemma:css}. Hence $\mc A_k \subset \mc R_k$. 

We prove by induction that $\mc R_k \subset \mc A_k$. First, $\mc R_1 \subset \mc A_1$ as $\mc E^* \subset \mc F$. For $k\geq 2$, assume that $\mc R_{k-1} \subset \mc A_{k-1}$. 
For each $P\in \mc H_k$, $\cl(P)$ is a closed path in $(\mc V, \mc F)$. If $\cl(P)$ is a cycle, then $R_{\cl(P)}\in \mc A_k$ by the definition of $\mc H_k$ and $\mc F$.
If $\cl(P)$ is not a cycle, then it can be decomposed into multiple cycles of length strictly less than $k$. Then $R_{\cl(P)} \in \mc R_{k-1}\subset \mc A_{k-1} \subset \mc A_k$ by induction.
Hence for both cases, $R_{\cl(P)} \in \mc A_k$. Last, by Lemma~\ref{lemma:css}, $\mc R_k \subset \mc A_k$. 
\end{IEEEproof}

Theorem~\ref{thm:k} shows that the largest cycle length in $(\mc V, \mc F)$ is a sufficient value of $k$ such that $\mc R_k = \mc R^{(\mc M_T, \mc E_T)}$, and this length is shorter than the largest cycle length in $(\mc M_T, \mc E_T)$.
Consequently, $(\mc V, \mc F)$ can be used for calculating the scheduling rate region with reduced computational complexity compared to using $(\mc M_T, \mc E_T)$. The complete rate region can be obtained by employing Johnson's algorithm~\cite{johnson1975finding} to
enumerate cycles in $(\mc V, \mc F)$.
To calculate the subset of the rate region $\mc R_k$, one can enumerate cycles in $(\mc V, \mc F)$ up to length $k$.
In the following subsection, we compare these approaches for calculating $\mc R_k$ with Algorithm~\ref{alg:region}.

\subsection{Numerical Evaluation}

We compare the different approaches for calculating the scheduling
rate region by numerical evaluations on the networks
$\mc N_{L,2}^{\text{line}}$ for $L=4$ to $11$. Since each network has
$D^*=1$, the only scheduling graph is $(\mc M_1, \mc E_1)$.
There are two main scenarios in the numerical evaluation:
\begin{enumerate}
\item Calculating the entire scheduling rate region by enumerating cycles:
  \begin{itemize}
  \item SR-1: Enumerate cycles in in original scheduling graph $(\mc M_1, \mc E_1)$.
  \item SR-2: Enumerate cycles in the reduced scheduling graph $(\mc V, \mc F)$.
  \end{itemize}
\item 
  Calculating a subset of the rate region $\mc R_k$:
  \begin{itemize}
  \item Rk-1: Enumerate cycles up to length $k$ in the original scheduling graph $(\mc M_1, \mc E_1)$.
  \item Rk-2: Enumerate cycles up to length $k$ in the reduced scheduling graph $(\mc V, \mc F)$.
  \item Rk-3: Use Algorithm~\ref{alg:region} to calculate $\mc R_k$.
  \end{itemize}
\end{enumerate}
All the approaches use $\mc N_{L,2}^{\text{line}}$ as the input. The operations of these approaches are summarized as follows: 
\begin{itemize}
\item SR-1: 
  \begin{enumerate}
  \item Evaluate $\mc M_2^*$ using the Bron–Kerbosch algorithm with
    vertex
    pivoting~\cite{bron1973algorithm,tomita2006worst,cazals2008note,eppstein2013listing}.
  \item Generate $(\mc M_1,\mc E_1)$ using Algorithm~\ref{alg:sg}. 
  \item Enumerate cycles of $(\mc M_1, \mc E_1)$ using Johnson's algorithm~\cite{johnson1975finding}.  
  \end{enumerate}
\item SR-2: 
  \begin{enumerate}
  \item Evaluate $\mc M_2^*$ using the Bron–Kerbosch algorithm with
    vertex pivoting.
  \item Generate $(\mc M^*_L,\mc M^*_R,\mc E^*)$ from $\mc M_2^*$. 
  \item Generate $(\mc V, \mc F)$ from $(\mc M^*_L,\mc M^*_R,\mc E^*)$.
  \item Enumerate cycles of $(\mc V, \mc F)$ using Johnson's algorithm.
  \end{enumerate}
\item Rk-1: %
  \begin{enumerate}
  \item Evaluate $\mc M_2^*$ using the Bron–Kerbosch algorithm with
    vertex
    pivoting.
  \item Generate $(\mc M_1,\mc E_1)$ using Algorithm~\ref{alg:sg}. 
  \item Enumerate cycles of $(\mc M_1, \mc E_1)$ up to length $4$.
  \end{enumerate}
\item Rk-2: %
  \begin{enumerate}
  \item Evaluate $\mc M_2^*$ using the Bron–Kerbosch algorithm with
    vertex pivoting.
  \item Generate $(\mc M^*_L,\mc M^*_R,\mc E^*)$ from $\mc M_2^*$. 
  \item Generate $(\mc V, \mc F)$ from $(\mc M^*_L,\mc M^*_R,\mc E^*)$.
  \item Enumerate cycles of $(\mc V, \mc F)$ up to length $4$.
  \end{enumerate}
\item Rk-3: %
  \begin{enumerate}
  \item Evaluate $\mc M_2^*$ using the Bron–Kerbosch algorithm with
    vertex pivoting.
  \item Generate $(\mc M^*_L,\mc M^*_R,\mc E^*)$ from $\mc M_2^*$. 
  \item Execute RateRegion in Algorithm~\ref{alg:region} with $k_{\max} = 4$.
  \end{enumerate}
\end{itemize}

In our evaluation, we implemented all these approaches using the Julia programming language. Johnson's algorithm and the algorithm for enumerating cycles up to a certain length were obtained from the Julia Graphs package~\cite{juliagraphs2021}. To measure the execution time accurately, we utilized the Julia BenchmarkTools package~\cite{juliabt2023}, which runs a function multiple times to obtain a more stable estimate of the running time.
Table~\ref{tab:com} presents a comparison of the execution times for these approaches. Based on our evaluation, we have made the following observations:

When calculating the entire rate region by enumerating all the cycles, SR-2 proves to be more efficient than SR-1. However,  the computational costs of both approaches increase rapidly as the network length $L$ increases. As a result, when $L$ reaches 6 for SR-1 and 7 for SR-2, the memory requirements become substantial, causing the program to crash on our computer. As a result, we were unable to obtain the running time for larger networks using these two approaches. This limitation highlights the challenge of enumerating all the cycles in larger networks, as the computational and memory requirements become increasingly demanding.

All three approaches, Rk-1, Rk-2, and Rk-3, are capable of calculating $\mc R_4$ for networks with $L=11$.
Among these approaches, Rk-2 is more efficient than Rk-1 for all the considered networks.
While Rk-2 outperforms Rk-3 for smaller networks, the advantage of Rk-3 becomes evident as the network size increases.
It is worth noting that when compared to Johnson's algorithm for enumerating all cycles, enumerating cycles up to a certain length tends to be slower when the length is relatively large~\cite{liu2006new}. This is the case for $\mc N_{4,2}^{\text{line}}$, where SR-1 is faster than Rk-1.

In conclusion, we would like to remark that $\mc R_4$ serves as the rate region, although we did not provide a formal proof in this paper. If evaluated, it would become apparent that $\mc R_4=\mc R_5=\mc R_6=\cdots$. However, relying solely on evaluating up to $\mc R_{|\mc V|}$ is not a feasible method to prove the rate region, as it would be impractical for larger networks.
In another work, we have demonstrated that $\mc R_4$ aligns with an upper bound on the rate region of $\mc N_{L,2}^{\text{line}}$, obtained by leveraging a graphical property of the periodic graph associated with the network.

\begin{table*}[tb]
  \centering
  \caption{Comparison of three methods for calculating the rate region of $\mc N_{L,2}^{\text{line}}$. All the methods are implemented in Julia, and executed on a computer with a 2.3 GHz Quad-Core CPU, 8 GB memory and Julia 1.9.2.} \label{tab:com}
  \begin{tabular}{crrrrrrrrrrrrrrrr}
    \toprule $ $
   Approach & $\mc N_{4,2}^{\text{line}}$ & $\mc N_{5,2}^{\text{line}}$& $\mc N_{6,2}^{\text{line}}$ & $\mc N_{7,2}^{\text{line}}$ & $\mc N_{8,2}^{\text{line}}$ & $\mc N_{9,2}^{\text{line}}$& $\mc N_{10,2}^{\text{line}}$ & $\mc N_{11,2}^{\text{line}}$\\ \midrule
    SR-1 & 0.928 ms &	5.538 s & - & -& -& -& -& -  \\
    SR-2 & 0.765 ms & 0.765 ms &	15.06 s & -& -& -& -& - \\
    \midrule
    Rk-1 & 0.981 ms &	4.321 ms &	21.38 ms &	131.5 ms &	794.1 ms &	3.518 s	& 13.16 s &	49.952 s \\
    Rk-2 & 0.685 ms &	0.857 ms &	2.850 ms &	15.93 ms &	127.6 ms &	466.7 ms &	2.939 s &  	8.462 s \\
    Rk-3 & 1.242 ms &	1.388 ms &	5.976 ms &	22.79 ms &	101.6 ms &	246.8 ms &	1.002 s &	3.204 s \\
    \bottomrule
  \end{tabular}
\end{table*}

\section{Algorithms for Maximizing a Linear Function of Rate Vectors}
\label{sec:algm}

In this section, our focus is on maximizing a linear function of rate vectors. This problem arises in network utility maximization, where the goal is to find a scheduling rate vector that maximizes a weighted sum (see,
e.g., \cite{jain2003impact,lin2006}). For instance, one common objective is to maximize the sum rate of all the links in the network. While it is technically possible to solve the maximization problem given the rate region, as we discussed earlier, evaluating the rate region can be challenging.
To address this issue, we propose an algorithm that maximizes a linear function without explicitly calculating the rate region. This approach allows us to find an optimal scheduling rate vector without relying on the explicit determination of the rate region.

Previous works have considered this optimization problem for
scheduling with delays~\cite{grokop2011interference,
  chitre2012throughput}. However, these works provide only
 approximate solutions to the optimization problem. 
They treat the step-$1$ scheduling graph
$(\mc M_T, \mc E_{T,1})$ as a state transition graph and employ a
dynamic programming approach similar to the Viterbi algorithm to
optimize the state sequence. Their objective is to find an optimal path (not cycle)
of a given length $k$ in the step-$1$ scheduling graph. In other
words, they find the optimal rate vector in $\dR^{\mc N^{T+k}}$, which is not necessarily equal to $\cR^{\mc N}$ even when $k=|\mc M_T|$. However, as $k$ tends to infinity, $\dR^{\mc N^{T+k}}$ converges to $\cR^{\mc N}$ (as discussed in Sec.~\ref{sec:framed}). 

In this section, we propose an approach to accurately compute the
optimal value of a linear function on the rate vectors using the step-$T$ scheduling graph $(\mc M_T, \mc E_{T})$ and the dominance property in Sec.~\ref{sec:sgcal}.
Our algorithm identifies an optimal cycle of a given length $k$ in $(\mc M_T, \mc E_{T})$ and has a computation cost that is linear in $k$.  

\subsection{Problem Formulation and Simplification}

Consider a network $\mc N= (\mc L, \mc I, D_{\mc L})$, and a linear
function $f:\mathbb{R}^{|\mc L|} \rightarrow \mathbb{R}$.
In \eqref{eq:rk}, we defined $\mc R_k$, which is the subset of the rate region generated by the cycles up to length $k$ in the step-$T$ scheduling graph $(\mc M_T,\mc E_T)$. In this context, we study how to calculate
the optimal value $\max_{R\in \mc R_k} f(R)$ and identify a
periodic schedule that achieves this optimal value.

As $f$ is linear, we can express it as
$f(R) = \sum_{l\in \mc L} c_l R(l)$, where $c_l$ represents fixed
linear combination coefficients. For our analysis, we will
specifically focus on maximizing a linear function $f$ in which all
the coefficients $c_l$ are positive. For linear functions with a
mixture of positive and negative coefficients, the problem can be
transformed into an equivalent problem with only positive
coefficients, as discussed below.

Consider a linear function $f$ with a coefficient $c_{l_0}\leq 0$. If $R^*\in \mc R_k$ maximizes $f(R)$, then the rate vector obtained by setting $R^*(l_0)$ to $0$ (which is also in $\mc R_k$ due to Lemma~\ref{lemma:dom}) is also optimal. This is because inactivating link $l_0$ generates no collision with other links.
Therefore, we can reduce the problem by removing link $l_0$ as follows: First let $(\mc L', \mc I')$ be the directed graph or hypergraph obtained by removing $l_0$ from $(\mc L, \mc I)$. In other words, $\mc L' = \mc L\setminus \{l_0\}$ and for $l\neq l_0$, $\mc I'(l)$ is obtained by excluding all $\theta\in \mc I(l)$ with $l_0\in \theta$. Second, let $D_{\mc L'}$ be the submatrix of $D_{\mc L}$ obtained by removing the row and the column indexed by $l_0$. Last, define $f'(R) = \sum_{l\in \mc L'} c_l R(l)$. The new optimization problem is to maximize $f'$ for the network $\mc N' \triangleq (\mc L', \mc I', D_{\mc L'})$.
By repeating this procedure, we can continue removing links with negative coefficients until $f'$ consists only of positive coefficients.

\subsection{An Incremental Approach for Optimizing a Linear Function}

Theorem~\ref{thm:uu} provides an approach for maximizing a linear function $f:\mathbb{R}^{|\mc L|} \rightarrow \mathbb{R}$ with positive coefficients. Specifically, since $f(R_1)\geq f(R_2)$ for any $R_1\pgeq R_2\in \mathbb{R}^{|\mc L|}$,  we know that $\max_{R\in \mc R_k} f(R) = \max_{i=1,\ldots,k} \max_{R\in \mc R_k^*} f(R)$. However, the complexity of this algorithm, however, can be exponential in $k$ due to the size of $\mc W_k(A,B)$. If we are only interested in the optimal value of $f$, we can simplify the algorithm by replacing the set $\mc W_k(A,B)$ with a real value. 

Define 
\begin{equation*}
   U_1^* = \max_{R\in \mc R_1} f(R)  = \max_{P\in \mc E^*}f(R_{\cl(P)}),
 \end{equation*}
 where the second equality follows from \eqref{eq:r1}.
For $A \in \mc M_L^*$ and $B \in \mc M_R^*$, define 
\begin{equation}\label{eq:u2ab}
    U_2(A,B) = \max \{f((B_1\land A_2)\cv 1): (A,B_1), (A_2,B)\in \mc E^* \},
\end{equation}
and for $k\geq 3$, define
\begin{equation}\label{eq:ukab}
  U_k(A,B) = \max \{U_{k-1}(A,B') + f((B'\land A')\cv 1): (A',B)\in \mc E^*, B'\in \mc M_R^* \}.
\end{equation}
The next theorem gives an incremental algorithm for optimization a linear function over $\mc R_k$. 
  
\begin{theorem}\label{thm:alr}
  Let $f:\mathbb{R}^{|\mc L|} \rightarrow \mathbb{R}$ be a linear function with positive linear combination coefficients.
  For $k\geq 1$, $\max_{R\in \mc R_k} f(R) = \max\{U_1^*, U_2^*, \ldots, U_k^*\}$, where for $i\geq 2$,
  \begin{equation}\label{eq:ustar}
    U_i^* = \frac{1}{iT} \max_{A\in \mc M_L^*,B\in \mc M_R^*} U_i(A,B) + f((A\land B)\cv 1).
  \end{equation}
\end{theorem}
\begin{IEEEproof}
  As the linear combination coefficients in $U$ are non-negative, by Theorem~\ref{thm:uu},
  \begin{equation*}
    \max_{R\in \mc R_k} f(R) = \max\left\{ \max_{R\in \mc R_1^*} f(R) = U_1^*, \max_{R\in \mc R_i^*} f(R), i=2,\ldots,k\right\},
  \end{equation*}
  where for $i\geq 2$,
\begin{IEEEeqnarray*}{rCl}
  \max_{R\in \mc R_i^*} f(R)
  & = & 
  \frac{1}{iT} \max_{R\in \bigcup_{A\in \mc M_L^*,B\in \mc M_R^*} \left({\mc R}_i(A,B)+(A\land B)\cv 1\right)} f(R) \\
  & = & \frac{1}{iT} \max_{A\in \mc M_L^*,B\in \mc M_R^*} \max_{R\in W_i(A,B)}f(R)+ f((A\land B)\cv 1).
\end{IEEEeqnarray*}
We show that $\max_{R\in \mc W_i(A,B)}f(R) = U_i(A,B)$ for $i\geq 2$ by induction.
First, by the definition in \eqref{eq:r2ab} and \eqref{eq:u2ab}, $U_2(A,B) = \max_{R\in \mc W_2(A,B)}f(R)$. 
For $i > 2$, assume that $\max_{R\in \mc R_{i-1}(A,B)}f(R) = U_{i-1}(A,B)$.
By the definition in \eqref{eq:rkab},
\begin{IEEEeqnarray*}{rCl}
  \max_{R\in \mc W_i(A,B)}f(R) &=& \max_{B'\in \mc M_R^*} \max_{A':(A',B)\in \mc E^* } \max_{R\in \mc R_{i-1}(A,B')} f(R) + f((B'\land A')\cv 1) \\
  & = & \max_{B'\in \mc M_R^*} \max_{A':(A',B)\in \mc E^* } U_{k-1}(A,B')+f((B'\land A')\cv 1) \\
  & = & U_k(A,B).
\end{IEEEeqnarray*}
The proof is completed by $\max_{R\in \mc R_i^*} f(R) = U_i^*$.
\end{IEEEproof}

In the subsequent subsections, algorithms are presented for computing the optimal value of the linear function $f(R)$ as well as identifying a cycle that attains this optimal value. 
 
\subsection{Algorithm for Optimal Value}

Algorithm~\ref{alg:uab} provides the pseudocode for calculating $\mc U_k$, and Algorithm~\ref{alg:rate} provides the pseudocode for determining the optimal value of $f(R)$ in $\mc R_k$.
The structure of Algorithm~\ref{alg:uab} and  Algorithm~\ref{alg:rate} remains similar to that of Algorithm~\ref{alg:rab} and Algorithm~\ref{alg:region}, respectively. The main difference lies in the computation of $U_k(A,B)$ instead of $\mc W_k(A,B)$.

Algorithm~\ref{alg:uab} and  Algorithm~\ref{alg:rate} assume that $(\mc M^*_L,\mc M^*_R,\mc E^*)$ has already been calculated (see Sec.\ref{sec:sgcal}). These two algorithms are described below. To simplify the notation, we denote
\begin{equation*}
   U_{k} = ( U_{k}(A,B),A\in \mc M_L^*, B\in \mc M_R^*).
\end{equation*}
In Algorithm~\ref{alg:uab}, two functions are provided: U2AB and UAB.
The U2AB function calculates $U_2$ using \eqref{eq:u2ab}, while the  UAB function calculates $ U_{k}$ from $ U_{k-1}$ using \eqref{eq:ukab}.
The computation cost of both U2AB and UAB is $O(|\mc M_L^*|^2|\mc M_R^*|^2|\mc L|T)$, accounting for the integer and logical operations involved in the calculation.

\begin{algorithm}[tb]
  \caption{The pseudocode for calculating $U_k$ consists of two functions: U2AB and UAB. The function U2AB is responsible for computing $U_2$, while the function UAB is used to calculate $U_{k}$ from $U_{k-1}$, and this process is applicable for any value of $k\geq 2$.} \label{alg:uab}
  \begin{algorithmic}[1]
    \Function{U2AB}{}
      \State \textbf{Input:} $\mc M^*_L, \mc M^*_R,\mc E^*$
      \State \textbf{Output:} $ U_{2}$ %
      \For{each $A\in \mc M_L^*$ and $B\in \mc M_R^*$}
        \State $ U_2(A, B) \leftarrow 0$
        \For{each $B_1$ s.t. $(A,B_1)\in \mc E^*$ and $A_2$ s.t. $(A_2,B)\in \mc E^*$}
          \State $R = f((B_1\land A_2)\cv 1)$
          \If{$R > U_2(A,B)$}
            \State $U_2(A, B) \leftarrow R$
          \EndIf
        \EndFor
      \EndFor
      \State return $U_{2}$ %
    \EndFunction
    \Function{UAB}{}
      \State \textbf{Input:} $ U_{k-1}, \mc M^*_L, \mc M^*_R,\mc E^*$
      \State \textbf{Output:} $ U_{k}$ %
      \For{each $A\in \mc M_L^*$ and $B\in \mc M_R^*$}
        \State $ U_k(A, B) \leftarrow 0$
        \For{each $B'\in \mc M_R^*$ and $A'\in \mc M_L^*$ s.t. $(A',B)\in \mc E^*$}
          \State $R = U_{k-1}(A,B') + f((B'\land A')\cv 1)$
          \If{$R > U_{k}(A,B)$}
            \State $U_k(A, B) \leftarrow R$
          \EndIf
        \EndFor
      \EndFor
      \State return $U_{k}$ %
    \EndFunction
  \end{algorithmic}
\end{algorithm}

In Algorithm~\ref{alg:rate}, two functions are provided: OptimalRate and UStar. The OptimalRate function takes an integer $k_{\max}$ as the input, and calculates $U_{k}^* $ for $k=1,2,\ldots,k_{\max}$ as the output. The OptimalRate function calls U2AB and UAB to obtain $U_{k}$ for $k=2,\dots,k_{\max}$ and then calculates $U_k^*$ by calling UStar on $U_{k}$, applying the formula in Theorem~\ref{thm:alr}.
The computation cost of UStar for $U_k^*$ is $O(|\mc M_L^*||\mc M_R^*||\mc L|T)$ floating point operations.
The overall computation cost of OptimalRate is $O(k_{\max}|\mc M_L^*|^2|\mc M_R^*|^2|\mc L|T)$. To get the optimal value over the entire rate region, it is sufficient to use $k_{\max} = |\mc M_L^*||\mc M_R^*|$, so that the computational complexity is  $O(|\mc M_L^*|^3|\mc M_R^*|^3|\mc L|T)$.

\begin{algorithm}[tb!]
  \caption{The pseudocode for calculating $U_k^*$ for $k=1,2,\ldots,k_{\max}$ consists of two functions: OptimalRate and UStar. The OptimalRate function is responsible for determining the optimal rate $f(R)$, while the UStar function  is called by OptimalRate to calculate $U_i^*$ from $U_i$.}\label{alg:rate}
  \begin{algorithmic}[1]
    \Function{OptimalRate}{}
      \State \textbf{Input:} $f$, $(\mc M^*_L, \mc M^*_R, \mc E^*)$ and integer $k_{\max}$
      \State \textbf{Output:} $\max_{R\in \mc R_k}f(R), k=2,\ldots,k_{\max}$
      \State $U_1^*\leftarrow 0$
      \For{each $A\in \mc M_L^*$ and $B\in \mc M_R^*$ s.t. $(A,B)\in \mc E^*$}
        \State $R \leftarrow \frac{1}{T}(B\land A)\cv 1$
        \If{$R>U_1^*$}
          \State $U_1^*\leftarrow R$
        \EndIf
      \EndFor
      \State $U_2 \leftarrow $ U2AB$(\mc M^*_L, \mc M^*_R,\mc E^*)$
      \State $U_2^* \leftarrow $ UStar$(U_2, \mc M^*_L, \mc M^*_R,2)$
      \For{$k$ from $3$ to $k_{\max}$}
        \State $U_k \leftarrow $ UAB$(U_{k-1},(\mc M^*_L, \mc M^*_R, \mc E^*))$
        \State $U_k^* \leftarrow $ UStar$(U_k, \mc M^*_L, \mc M^*_R,k)$
      \EndFor
      \State return $U_k^*, k=2,\ldots,k_{\max}$%
    \EndFunction
    \Function{UStar}{}
      \State \textbf{Input:} $U_k$, $\mc M^*_L$, $\mc M^*_R$, $k$
      \State \textbf{Output:} $\max_{R\in \mc R_k}f(R)$
      \State $U_k^*\leftarrow 0$
      \For{each $A\in \mc M_L^*$ and $B\in \mc M_R^*$}
        \State $R = \frac{1}{kT}(U_k(A,B)+f((A\land B)\cv 1))$
        \If{$R > U_k^*$}
          \State $U_k^*\leftarrow R$
        \EndIf
      \EndFor
      return $U_k^*$
    \EndFunction
  \end{algorithmic}
\end{algorithm}

\subsection{Algorithm for Optimizer}

In addition to determining the optimal value $U_k^*$, it is also
essential to identify the optimizer, which refers to the cycle that
achieves the optimal value.  Such a cycle can be utilized to construct
an optimal periodic schedule.

Assume that $U_k^*$ and the $U_k$ are already calculated for $k=1,\ldots,k_{\max}$, which can be done by Algorithm~\ref{alg:uab} and  Algorithm~\ref{alg:rate}. 
First, 
the optimizer of $U_1^*$ can be find by enumerating all elements $P$ in $\mc E^*$ until we find $P$ such that $f(R_{\cl(P)}) = U_1^*$.

The case with $k\geq 2$ can be solved using backward searching. Enumerate $A_1\in \mc M_L^*$ and $B_k\in \mc M_R^*$ until we find $A_1$ and $B_k$ such that 
\begin{equation*}
  U_k^* = U_k(A_1,B_k) + f((A_1\land B_k)\cv 1).
\end{equation*}
The existence of such $A_1$ and $B_k$ is guaranteed by the definition of $U_k^*$ (see \eqref{eq:ustar}).
Then for $i=k, k-1,\ldots,3$, enumerate $A_i\in \mc M_L^*$ and $B_{i-1}\in \mc M_R^*$ until we find $A_i$ and $B_{i-1}$ such that $(A_i,B_i)\in \mc E^*$ and
\begin{equation*}
  U_i(A_1,B_i) = U_{i-1}(A_1,B_{i-1}) + f((B_{i-1}\land A_{i})\cv 1).
\end{equation*}
The existence of such $A_i$ and $B_{i-1}$ is guaranteed by the definition of $U_i(A,B)$ (see \eqref{eq:ukab}).
Last, enumerate  $B_1\in \mc M_R^*$ and $A_2\in \mc M_L^*$ such that $(A_1,B_1), (A_2,B_2)\in \mc E^*$ and 
\begin{equation*}
  U_2(A_1,B_2) = f((B_1\land A_2)\cv 1).
\end{equation*}
The existence of such $A_2$ and $B_{1}$ is guaranteed by the definition of $U_2(A,B)$ (see \eqref{eq:u2ab}).

According to the above construction, $(A_1,B_1),\ldots,(A_k,B_k) \in \mc E^*$ and
\begin{equation*}
  U_k^* = f((A_1\land B_k)\cv 1) + \sum_{i=3}^kf((B_{i-1}\land A_{i})\cv 1) + f((B_1\land A_2)\cv 1).
\end{equation*}
Therefore, $(B_1\land A_2,B_{2}\land A_3,\ldots,B_{k-1}\land A_k, B_k\land A_1,B_1\land A_2)$ is a  $k$-cycle in $(\mc M_T, \mc E_T)$ that achieves the optimal value $U_k^*$.

We give the pseudocode of this backward searching algorithm in Algorithm~\ref{alg:bs}, where a function called OptimalCycle is provided.
The whole procedure is as follows:
\begin{enumerate}
\item Calculate  $U_k^*$ and the $U_k$ for $k=1,\ldots,k_{\max}$ by Algorithm~\ref{alg:uab} and  Algorithm~\ref{alg:rate}.
\item Determine $k^*$ such that $U_{k^*}^* = \max_{k=1,\ldots,k_{\max}} U_k^*$.
\item Execute the OptimalCycle function to identify an optimizer for $U_{k^*}^*$. 
\end{enumerate}
The extra computation cost of OptimalCycle is $O(k^*|\mc M_L^*||\mc M_R^*||\mc L|T)$.

\begin{algorithm}[tb!]
  \caption{The pseudocode for identifying a cycle that achieves $U_k^*$ where $k\geq 2$. The values of $U_i$, $i=2,\ldots,k$ are known. The output is a cycle $C$ such that $U_k^* = f(R_C)$.}\label{alg:bs}
  \begin{algorithmic}[1]
    \Function{OptimalCycle}{}
      \State \textbf{Input:} $U_k^*$, $k$
      \State \textbf{Output:} $C = (C_1,\ldots,C_k,C_1)$
      \For{each $A\in \mc M_L^*$ and $B\in \mc M_R^*$}
        \If{$U_k(A,B)+f(A\land B)\cv 1)=U_k^*$}
          \State $A_1 \leftarrow A$, $B_k\leftarrow B$
          \State $C_k \leftarrow B_k\land A_1$
          \State break
        \EndIf
      \EndFor
      \For{$i$ from $k$ down to $3$}
        \For{each $A$ s.t. $(A,B_i)\in \mc E^*$ and $B\in \mc M_R^*$}
          \If{$U_{i-1}(A_1,B)+f(B\land A)\cv 1)=U_i(A_1,B_i)$}
            \State $A_i \leftarrow A$, $B_{i-1}\leftarrow B$
            \State $C_{i-1} \leftarrow B_{i-1}\land A_i$
            \State break
          \EndIf
        \EndFor
      \EndFor
      \For{each $A$ and $B$ s.t. $(A,B_2), (A_1,B)\in \mc E^*$}
        \If{$f(A\land B)\cv 1)=U_2(A_1,B_2)$}
          \State $A_2 \leftarrow A$, $B_1\leftarrow B$
          \State $C_1 \leftarrow B_1\land A_2$
          \State break
        \EndIf
      \EndFor
      \State return $(C_1,\ldots,C_k,C_1)$
    \EndFunction
  \end{algorithmic}
\end{algorithm}

\section{Concluding Remarks}
\label{sec:con}

This work introduces a graphical framework for wireless network
scheduling with discrete signal propagation delays. It extends the
existing independent set-based scheduling framework, commonly used in
traditional scheduling with guard intervals to prevent collisions. To
gain a better understanding of the advantages and the feasibility of
scheduling with delays in the real world, several further research
directions can be explored:
\begin{enumerate}
\item Outer bounds on the scheduling rate region: Due to the high computational cost involved in calculating the complete scheduling rate region, it may only be feasible to compute a subset of it in practical cases. Evaluating the quality of the computed subset can be facilitated by establishing an outer bound on the rate region.
  
\item Practical scheduling approaches: The algorithms proposed in this paper make ideal assumptions, such as assuming synchronization of all network nodes to a common clock and having complete and accurate delay and collision information. Further research is needed to relax these assumptions and develop practical scheduling approaches that can handle real-world network scenarios.
  
\item Network flow control: The presented framework opens up possibilities for systematically studying end-to-end communication flows in wireless networks with delays. Scheduling with delays should be jointly optimized with routing and congestion control mechanisms to ensure efficient and reliable network operation. 

\item Real-world demonstrations: Conducting practical experiments and demonstrations can help assess the performance and practical advantages of such scheduling techniques in real-world wireless network environments. This includes evaluating the impact on throughput, latency, energy efficiency, and overall system performance.
\end{enumerate}

\appendix[Physical Network Model]
\label{sec:phy}

We introduce a physical model of wireless networks in~\cite{tse2005fundamentals}, and discuss how to apply the results for our network model to the physical model.
We denote this physical network model as $\mc N^{\mathrm{phy}}$, which 
has the $N$ nodes that share the same communication channel of bandwidth $W$.
Denote by $P_i$ the transmitting power of node $i$, and by $h_{ij}$ the channel gain from the node $i$ to the node $j$, where $1\leq i \neq j \leq N$. So when the node $i$ transmits, the node $j$ can receive the signal power $h_{ij}P_i$. 
Denote by $R_{ij}^{\mathrm{code}}$ the coding rate from the node $i$ to the node $j$, where $1\leq i \neq j \leq N$. Here we assume $h_{ij}$, $R_{ij}^{\mathrm{code}}$ and $P_i$ do not change over time.

Based on the physical model $\mc N^{\mathrm{phy}}$, we can derive a network model $\mc N$ in~Sec.~\ref{sec:dnm}. The link set $\mc L = \{l_{ij}\triangleq(i,j): R_{ij}^{\mathrm{code}}>0\}$. For any $l_{ij}\in \mc L$ and $\theta\subset \mc L$, we say $\theta$ is in the collision set $\mc I(l_{ij})$ if
\begin{equation*}
  \frac{1}{2} \log \left( 1 + \mathrm{SINR} \right) \leq R_{ij}^{\mathrm{code}}, 
\end{equation*}
where the signal-to-interference-and-noise ratio $\mathrm{SINR} = \frac{h_{ij}P_{i}}{\sum_{l \in \theta} h_{\tx_l j}P_{\tx_l}+N_0W}$, and $N_0$ is the power spectral density of the white noise process. $\mc N^{\mathrm{phy}}$ and $\mc N$ share the same delay matrix.

If a collision-free schedule $S$ of $\mc N$ is applied to $\mc N^{\mathrm{phy}}$, for each link $(i,j)\in \mc L$, rate $R_{ij}^{\mathrm{code}}$ can be achieved for any active timeslot. Hence, we obtain an achievable rate vector $(R_{ij}, 1\leq i \neq j \leq N)$ for $\mc N^{\mathrm{phy}}$, where
\begin{equation*}
  R_{ij} =
  \begin{cases}
    R_{ij}^{\mathrm{code}}R^{\mc N}_{S}(l_{ij}), & l_{ij} \in \mc L, \\
    0, & \text{otherwise}.
  \end{cases}
\end{equation*}
Therefore, the rate region of ${\mc N}$ induces an achievable rate region of 
$\mc N^{\mathrm{phy}}$.

\end{document}